\def\llncs{0}
\def\fullpage{1}
\def\anonymous{0}
\def\draft{1}
\def\submission{0}

\ifnum\submission=1
\def\llncs{1}
\def\draft{0}
\def\anonymous{1}
\def\fullpage{0}
\fi

\ifnum\llncs=1
    \documentclass{llncs}
    \ifnum\fullpage=1
    \usepackage{fullpage}
    \fi
\else
    \documentclass[letterpaper]{article}
    \ifnum\fullpage=1
    \usepackage{fullpage}
    \fi
    \usepackage{microtype}
\fi

\usepackage{verbatim}
\usepackage{authblk}
\usepackage{graphicx} 
\IfFileExists{physics.sty}{%
  \usepackage{physics}%
}{%

  \providecommand{\norm}[1]{\left\lVert ##1\right\rVert}

}
\usepackage{xcolor} 
\usepackage{amsmath}
\allowdisplaybreaks[4] 
\usepackage{amssymb}
\usepackage{mathtools}

 \usepackage{amsthm}
 \usepackage{thmtools}
\usepackage{enumitem} 
\IfFileExists{dsfont.sty}{%
  \usepackage{dsfont}%
}{%
}
\usepackage[colorlinks=true,linkcolor=magenta,citecolor=blue,pagebackref=true,hypertexnames=false,pdftex,pdfpagelabels,bookmarks,hyperindex,hyperfigures]{hyperref}
\usepackage{dashbox}
\usepackage{fancybox,framed}
\usepackage[skins]{tcolorbox}
\usepackage{subcaption}
\IfFileExists{breakcites.sty}{\usepackage{breakcites}}{}
\usepackage{comment} 
\let\subparagraph\paragraph
\usepackage{titlesec}
\titleformat*{\paragraph}{\normalsize\bfseries}

\usepackage[capitalise,nameinlink]{cleveref}

\ifnum\draft=1
	\newcommand{\minki}[1]{\textcolor{blue}{$\langle\langle$Minki: #1$\rangle\rangle$}}
	\newcommand{\yixin}[1]{\textcolor{red}{$\langle\langle$Yixin: #1$\rangle\rangle$}}
\else
	\newcommand{\minki}[1]{}
	\newcommand{\yixin}[1]{}
\fi

\ifnum\llncs=0
	\newtheorem{theorem}{Theorem}[section]
	\newtheorem{lemma}[theorem]{Lemma}
	\newtheorem{corollary}[theorem]{Corollary}
	
	\newtheorem{definition}[theorem]{Definition}

	\newtheorem{claim}[theorem]{Claim}

	\theoremstyle{remark}
	\newtheorem{remark}[theorem]{Remark}

\else
	\spnewtheorem{algorithm}{Algorithm}{\bfseries}{\rmfamily}

	\spnewtheorem{claim}{Claim}{\bfseries}{\itshape}
\fi
\crefname{appendix}{Appendix}{Appendices}
\Crefname{appendix}{Appendix}{Appendices}

\newcommand{\cA}{\mathcal{A}}

\newcommand{\N}{\mathbb{N}}

\newcommand{\cN}{\mathcal{N}}

\newcommand{\eps}{\varepsilon}

\renewcommand{\epsilon}{\varepsilon}

\newcommand{\R}{\mathbb{R}}
\newcommand{\Z}{\mathbb{Z}}
\newcommand{\F}{\mathbb{F}}
\newcommand{\E}{\mathbb{E}}
\newcommand{\Prob}{\mathbb{P}}

\newcommand{\cG}{\mathcal{G}}

\newcommand{\cK}{\mathcal{K}}
\newcommand{\cL}{\mathcal{L}}
\newcommand{\cR}{\mathcal{R}}

\providecommand{\ip}[2]{}
\renewcommand{\ip}[2]{\left\langle #1,#2\right\rangle}
\newcommand{\opnorm}[1]{\left\lVert #1\right\rVert_{\mathrm{op}}}
\newcommand{\Frob}[1]{\left\lVert #1\right\rVert_{\mathrm{F}}}
\newcommand{\dist}{\operatorname{dist}}
\newcommand{\poly}{\operatorname{poly}}
\newcommand{\polylog}{\operatorname{polylog}}
\newcommand{\negl}{\operatorname{negl}}
\newcommand{\ind}{\mathbf{1}}
\newcommand{\Id}{I_n}

\newcommand{\SVP}{\mathsf{SVP}}
\newcommand{\CVP}{\mathsf{CVP}}
\newcommand{\BDD}{\mathsf{BDD}}
\newcommand{\DGS}{\mathsf{DGS}}

\DeclareMathOperator{\GL}{GL}

\ifnum\draft=1
	
\else
	
\fi

\ifnum\llncs=1
  \spnewtheorem{convention}[theorem]{Convention}{\bfseries}{\itshape}
\else
  
\fi

\crefname{algorithm}{Algorithm}{Algorithms}
\Crefname{algorithm}{Algorithm}{Algorithms}
\newenvironment{algorithmblock}[1]{%
  \refstepcounter{algorithm}%
  \par\addvspace{0.8em}%
  \noindent\begin{minipage}{\linewidth}%
  \hrule\smallskip
  \textbf{Algorithm \thealgorithm: #1}\par\smallskip
  \begin{enumerate}[
    leftmargin=1.62em,
    label=\arabic*.,
    itemsep=0.13em,
    topsep=0.15em
  ]
}{%
  \end{enumerate}\smallskip\hrule
  \end{minipage}\par\addvspace{0.8em}%
}

\title{Solving the Shortest Vector Problem in $2^{0.6039n}$ Time \\via Mid-point Hessian}
\hypersetup{
  pdftitle={Solving the Shortest Vector Problem in time 2^{0.6039n} via Mid-point Hessian},
  pdfauthor={Minki Hhan}
}
\date{}

\ifnum\llncs=1
    \ifnum\anonymous=1
        \author{}
        \institute{}
    \else
        \author{
        Minki Hhan\inst{1} 
        }
        \institute{
        KAIST, Daejeon, Korea
        \\\email{minkihhan@kaist.ac.kr}
        }
    \fi
\else
    \author{Minki Hhan}
    \affil{{\small KAIST, Daejeon, Korea}
    \authorcr{\small minkihhan@kaist.ac.kr}}
    
\fi

\begin{document}

\maketitle

\begin{abstract}
We present randomized algorithms for the shortest vector problem (SVP). For the $n$-dimensional lattice $\mathcal L$, our algorithms solve SVP in time $2^{0.6039n+o(n)}$ classically and $2^{0.5411n+o(n)}$ quantumly and space $2^{0.5n+o(n)}$, improving the previous best algorithm running in $2^{n+o(n)}$ time and space of Aggarwal, Dadush, Regev, and Stephens-Davidowitz [STOC'15].

Our algorithms heavily use the property of the Hessian of the periodic Gaussian function at the half shortest vector: For a shortest vector $v \in \mathcal L$, the Hessian at $v/2$ has the eigenvector close to $v$, which can be used to recover $v$ using the (preprocessing) bounded distance decoding algorithm. Given the periodicity modulo $\mathcal L$, the candidate midpoints are indexed by the parity classes in $\mathcal L/2\mathcal L$. Our algorithm searches for the class of a shortest vector by estimating the corresponding Hessians using discrete Gaussian samples. 

We optimize the algorithm using random sublattice cosets and various sampling technique, achieving the final complexity. The optimization techniques may be of independent interest.

\end{abstract}

\vfill
\paragraph{AI use disclosure.} Most results of this paper are discovered with the assistance of ChatGPT 5.5 Pro and ChatGPT 5.6 Sol Ultra. While the author suggested several directions and optimizations, all the technical details are discovered by AI, and verified by the author. The manuscript is written by the author based on the initial draft generated by AI. The author takes full responsibility for the content. All errors in the manuscript are likely due to the author's human errors.
\clearpage

    \newpage
    \setcounter{tocdepth}{2}
    \tableofcontents
    \newpage

\section{Introduction}\label{sec:introduction}

An $n$-dimensional lattice
$\cL=\cL(b_1,\ldots,b_n)$ is the set of all integer combinations of
linearly independent vectors $b_1,\ldots,b_k\in\R^n$. For the sake of simplicity, we only consider the full rank lattice satisfies $k=n$.
The (search) shortest vector problem ($\SVP$) asks to find a nonzero vector of
minimum norm, given a basis of $\cL$. 

There are two faces of lattice problems including $\SVP$. 
Approximate variants of $\SVP$ have long served as algorithmic tools in computational number theory, integer programming, and cryptanalysis \cite{LLL82,STOC:Kannan83,Len83,FOCS:Shamir82,LO85}. On the other hand, 
the exact problem and its small approximation factor version are known to be hard \cite{STOC:Ajtai98,STOC:BCGR23,STOC:AggSte18,JACM:Khot05,STOC:HavReg07}.
Given the close relation between the lattice problems and many recent cryptographic primitives \cite{JACM:Regev09,STOC:BLPRS13,FOCS:MicReg04,C:Regev06,STOC:AjtDwo97,STOC:GenPeiVai08},
lattice problems have become one of the most promising foundations for post-quantum cryptography.

$\SVP$ remains a fundamental benchmark for our understanding of the complexity of
lattice problems. Researchers have discovered various worst-case algorithms \cite{C:HanSte07,STOC:MicVou10,STOC:AjtKumSiv01,STOC:Kannan83,NguVid08,EPRINT:PujSte09,SODA:MicVou10,CCL18} and 
heuristic or average-case algorithms \cite{NguVid08,EC:PouShe26,SODA:BDGL16}. 
To date, however, the best provable
worst-case running time was $2^{n+o(n)}$, with the same space
complexity~\cite{STOC:ADRS15} using a new efficient discrete Gaussian sampling algorithm.
Later, \cite{SICOMP:ACKS25} gave several time-space tradeoffs and quantum speedup,
including $2^{1.669n+o(n)}$ time and $2^{n/2+o(n)}$ space and $2^{0.950n+o(n)}$ and $2^{0.835n+o(n)}$ quantum time without and with quantum random access memory (QRAM), respectively.

We give a new class of worst-case algorithms for
$\SVP$, significantly improving the previous classical and quantum algorithms \cite{STOC:ADRS15,SICOMP:ACKS25}. Our main results are summarized as follows.

\begin{theorem}\label{thm:intro-main}
There are randomized classical and quantum algorithms
that solve Search-$\SVP$ with probability at least $2/3$ in expected
time
$  2^{0.603867n+o(n)}$ classically and $2^{0.541051n+o(n)}$ quantumly
and space
$  2^{0.5n+o(n)},$
up to factors polynomial in the input length. The quantum algorithm requires a QRAM of size $2^{0.36036n+o(n)}$.
\end{theorem}

\begin{remark}[Concurrent works]
    During the preparation of this manuscript, OpenAI announced an improvement of the sphere packing bound \cite{OpenAI2026TenAdvances}. Plugging that parameter, the time complexity of our algorithm is reduced to $2^{0.60221n + o(n)}$ classically and $2^{0.53925n+o(n)}$ quantumly, according to ChatGPT 5.6 Sol's calculation. We did not present these complexities because the result of OpenAI and the numerical calculation for those parameters are not verified (by the author, at least). We also noticed the concurrent works \cite{cryptoeprint:2026/1587,Kim26} right before publicizing the manuscript, which obtains $2^{(0.5+\frac{\beta^2}{4e\ln 2})n+o(n)}\approx 2^{0.7134n+o(n)}$ time and $2^{0.5n+o(n)}$ space algorithm for $\beta = 2^{0.4014\ldots}$, matching \cref{thm:direct-row}, with a completely different approach.
\end{remark}

We discuss some direct applications of \cref{thm:intro-main}. Obviously, our algorithm solves the approximate or unique $\SVP$ as well as its gap version in the same time and space complexity. Our algorithm beats the previous best time records 
\cite{RSA:WeiLiuWan15,EC:AggLiSte21,cryptoeprint:LWXZ11} $2^{0.802n+o(n)}$ time and
$2^{0.401n+o(n)}$ space for the approximate $\SVP$ with a constant factor, although they use less space.

This improvement can be readily used in other applications. For example, the best-known polynomial factor approximate $\SVP$ \cite[Theorems 5.3]{EC:AggLiSte21} uses one call to a constant-factor $\SVP$ oracle in a small dimension. Originally the proof uses the oracle from \cite{cryptoeprint:LWXZ11} with an exponent $0.802n$. Replacing this to ours, we improve the time complexity for solving $\widetilde O (n^c)$-$\SVP$ from $2^{\frac{n}{2c+1.24}}$ to $2^{\frac{n}{2c+1.65}}$.

Another direct application is $\BDD$ or exact $\CVP$ with the distance promise. Using Kannan's embedding, $\CVP$ for $(t,\cL)$ with the promise $\dist(t,\cL) < \sqrt3 \lambda_1(\cL)/2$ can be reduced to $\SVP$ with one additional dimension (see e.g., \cite{IPL:LiuWanXuZhe14}). Thus $\BDD$ and the $\CVP$ with the above promise can be solved in the same time and space using our algorithm. Previously, \cite{CCC:DRS14,STOC:ADRS15} give $2^{n/2+o(n)}$ time algorithm for $\alpha<0.422$.

Lastly, using the dimension-preserving reduction from centered $\DGS$ to $\SVP$
\cite{SODA:StephensDavidowitz16}, we can sample one (or polynomially many) discrete Gaussian sample with an arbitrary parameter in the same time. This partly answers the open question in
\cite{STOC:ADRS15} for the centered discrete Gaussian sampling below smoothing, albeit with a worse time complexity than $2^{n/2+o(n)}$. A concurrent work \cite{Kim26} gives a $2^{n/2+o(n)}$ time discrete Gaussian sampling algorithm for arbitrary parameters.

\subsection{Technical overview}
Let $\cL$ be a $n$-dimensional full-rank lattice with basis $B$ and $s>0$ be a width parameter. 
Its dual lattice $\cL^*=\{y: \ip{y}{x} \in \Z \text{ for every }x \in \cL\}$ has the basis $B^{-T}.$
Define 
$\rho_s(x):=\exp\!\left(-\pi\frac{\norm{x}^2}{s^2}\right)$ and let $\rho_s(A):=\sum_{x\in A}\rho_s(x).$ 
We define the (centered) discrete Gaussian distribution
$D_{\cL,s}(x)
  :=\frac{\rho_s(x)}{\rho_s(\cL)}.$
The periodic Gaussian function \cite{JACM:AhaReg05} $F_s : \R^n \to \R$ is defined by
\[
  F_s(z)
  :=
  \frac{\rho_s(\cL+z)}
       {\rho_s(\cL)}
       =
  \E_{X\sim D_{\cL^*,1/s}}
  \left[
  e^{2\pi i \langle X,z\rangle}
  \right]
\]
where the last equality is due to the Poisson summation formula. It is clear that $F_s$ is $\cL$-periodic, i.e., $F_s(z)=F_s(z+v)$ for $v\in \cL$.

\paragraph{At the midpoint, the Hessian points the way.}
Our algorithm starts with studying its Hessian \cite{RegSte17,CCC:DRS14}, which also has two representations
\begin{align}\label{eqn: intro_Hessian_eigen}
        \nabla^2F_s(z)+\frac{2\pi}{s^2}F_s(z)\Id
  =
  \frac{4\pi^2}{s^4\rho_s(\cL)}
  \sum_{y\in\cL}
  (y-z)(y-z)^T
  e^{-\pi\norm{y-z}^2/s^2}
\end{align}
by differentiating the first representation, and by differentiating the second representation
\begin{align}\label{eqn: intro_Hessian_approx}
    \nabla^2F_s(z)
  =
  -4\pi^2
  \E_{X\sim D_{\cL^*,1/s}}
  \left[
    XX^T e^{2\pi i\ip{X}{z}}
  \right].
\end{align}
From these two representations, we observe the following two properties. 
First, let $v$ be a shortest vector with $\norm{v}=\lambda$ and evaluate $\nabla^2 F_s(v/2)$ in \cref{eqn: intro_Hessian_eigen}, we can see that
\[
\nabla^2F_s(v/2) + aI = \frac{4\pi^2}{s^4 \rho_s(\cL)} \left(\frac{v v^T}{2}e^{-\pi \lambda^2/4s^2 } + \sum_{y \in \cL\setminus \{0,v\} } (y-v/2)(y-v/2)^T e^{-\pi \norm{y-v/2}^2/s^2}\right)
\]
where $a:=2\pi F_s(v/2)/s^2$ which is unimportant in the spectral analysis. Noting that $\norm{y-v/2}$ for $y\neq 0,v$ is much larger than $\norm{v/2}=\lambda/2$, we can expect that $\nabla^2F_s(v/2) + aI$ is close to a multiple of $vv^T$. For appropriately chosen $s$, this intuition can be formalized by the following statement for the normalized eigenvector $\tilde v$ of $\nabla^2F_s(v/2)$ corresponding to the largest eigenvalue:
\begin{center}
    $\tilde v$ is \emph{inverse polynomially close} to $v/\norm{v}$. 
\end{center}
Then, given $\tilde v$, we guess $\lambda$ (which only takes polynomial iterations) and solve the bounded distance decoding ($\BDD$) problem for the target vector $\lambda \tilde v$. This $\BDD$ step can be relatively quick, thanks to the preprocessing BDD algorithm in \cite{SICOMP:ACKS25}, which solves each BDD instance with inverse polynomial distances in time $2^{o(n)}$ after $2^{0.5n+o(n)}$ time preprocessing.

\paragraph{One Hessian is enough if we know where to look.}
This gives our very first result (\cref{thm:direct-hessian-svp}) with time complexity $2^{1.463n+o(n)}$ and space complexity $2^{0.5n+o(n)}$.
Suppose that for a shortest vector $v$, we know a representative $w$ of the parity class $v+2\cL$.
We \emph{compute} the eigenvector $\tilde w$ of $\nabla^2F_s(w/2)$, and solve the BDD problem for it using (preprocessing) BDD. But, how to compute $\nabla^2F_s(w/2)$?

Here the second representation \cref{eqn: intro_Hessian_approx} comes in. We use the algorithm in \cite{STOC:ADRS15} that samples $N=2^{n/2}$ elements, say $X_1,...,X_{N}$, from the discrete Gaussian distribution in time $2^{n/2+o(n)}$ on $\cL^*$. Then, given $w/2$, we use the estimator
\[
\widehat{\nabla^2 F_s}(w/2) = 
  \frac{-4\pi^2}{N}
  \sum_{i=1}^N X_i X_i^T e^{2\pi i\ip{X_i}{w/2}}.
\]
Concentration inequalities and results from lattice literature show that about $2^{2t_0n+o(n)}$ samples for $t_0 = 2^{0.802}/4e \ln 2 = 0.2314$ suffice to accurately estimate the actual $\nabla^2 F_s(w/2)$ as well as its eigenvectors. This leads to the time complexity $|\cL/2\cL| \cdot 2^{2t_0n} =2^{1.463n +o(n)}$ by enumerating all the parity class in $\cL/2\cL$. The $2^{n/2+o(n)}$ space is required in the discrete Gaussian sampling.

\paragraph{Walsh-Hadamard looks everywhere at once.}
The time can be reduced further. Write $w=Bu$.
Since $B^TX\in\Z^n$,
\[
  e^{2\pi i\ip{X}{Bu/2}}
  =
  e^{\pi i u^TB^TX}
  =
  (-1)^{u^T(B^TX\bmod2)}.
\]
With some work, the fast Walsh-Hadamard transform can be applied to compute $\widehat{\nabla^2 F_s}(Bu/2)$ simultaneously for different $u$'s. Carefully choosing the coordinates leads to $2^{n+o(n)}$ time algorithm in \cref{thm:full-scan} while preserving $2^{0.5n+o(n)}$ space.

\paragraph{A random coset narrows the view.}
The bottleneck $2^n$ of the above strategy is the size of the search space $\cL/2\cL$.
Let $X\in \cL^*$.
Choose a random invertible linear map $P\in\GL_n(\F_2)$ and divide
\[
  P(B^TX\bmod 2)\in\F_2^n
\]
into $h$ fixed coordinates and $\ell=n-h$ remaining coordinates.
We then choose a random value $j\in\F_2^h$ for the first $h$ coordinates,
and use the discrete Gaussian samples $X$ that agree to this value. These samples belong to an affine coset of an index-$2^h$ sublattice of $\cL^*$. It turns out that the discrete Gaussian samples (for an appropriate width) are almost uniformly distributed over different cosets, thus this collective sampling adds $2^h$ multiplicative factor in the time complexity in sampling. That is, the estimation can be anyway done.

At first sight, it may seem that choosing one random coset could lose the spectral property of Hessian.
This is not the case. Conditioning the first $h$ coordinates of the samples on the dual lattice side corresponds to \emph{summing} over the corresponding $h$ coordinates of $u$ as shown in \cref{eqn: cGdual}. 
Thus, the shortest-vector contribution still appears in one of the $2^\ell$ coset Hessians!
Given this intuition, we can indeed apply almost the same strategy to find a shortest vector.
The overall time complexity becomes
\[
2^{2t_0n+h+o(n)} + 2^{\ell +o(n)}
\]
where the first term is to sample $2^{2t_0n}$ discrete Gaussian samples in the target coset (as we get roughly one coset Gaussian sample out of $2^h$ lattice Gaussian samples), and $2^\ell$ is the number of outputs to be checked. Balancing these costs gives $2^{0.73147n+o(n)}$ algorithm as shown in \cref{thm:direct-row}.

\paragraph{May the fourth weigh what matters.}
The main loss in the above algorithm is due to the inefficient rejection step in the discrete Gaussian sampling over the sublattice coset. The last algorithm uses the following observation: For two probabilistic distributions $D$ and $E$, where we know how to sample only from $E$ and know their probabilistic density functions $p$ and $q$, then we have
\[
\E_{X \sim D}[f(X)] = \sum_X p(X)f(X) = \sum_X q(X)\cdot \frac{p(X)}{q(X)}f(X) = \E_{X \sim E}\left[\frac{p(X)}{q(X)} f(X)\right]
\]
so that we can estimate the statistics regarding $D$ only using samples from $E$. We use this idea.

We first show that for a slightly wider width $s_R$ than original, we can sample discrete Gaussian directly on the target sublattice coset $\Lambda$.
Of course, samples from the discrete Gaussian with a wider width $s_R$ do not have exactly the distribution required for Hessian estimation. We therefore use the above idea and assign each sample $X$ the importance weight 
for the narrower target width $s_r<s_R$
\[
  w(X)
  :=
  \frac{\rho_{s_r}(X)}
       {\rho_{s_R}(X)} = \frac{\Pr_{D_{\Lambda,s_r}}[X]\rho_{s_r}(\Lambda)}{\Pr_{D_{\Lambda,s_R}}[X]\rho_{s_R}(\Lambda)} \propto \frac{\Pr_{D_{\Lambda,s_r}}[X]}{\Pr_{D_{\Lambda,s_R}}[X]}, 
\]
which is easily computable.
Similarly to the above, we have
\[
  \E_{X\sim D_{\Lambda,s_r}}[f(X)]
  =
  \frac{
    \E_{X\sim D_{\Lambda,s_R}}
    \left[
      w(X)f(X)
    \right]
  }{
    \E_{X\sim D_{\Lambda,s_R}}
    \left[
      w(X)
    \right]
  }
\]
for a function $f$.
Thus the wider Gaussian with width $R$ is used only to generate the samples, while the weights recover the Hessian corresponding to the narrower target Gaussian with width $r$.

The loss caused by importance sampling is
$  \iota=
  \frac12
  \log_2
  \left(
    \frac{R^2}{r(2R-r)}
  \right)$ in the exponent of the variance (\cref{lem:importance-weight-bounds}), forcing us to sample $2^{\iota n}$ more samples.
This turns $2^{2t_0n+h}$ time in the previous algorithm to $2^{(2r+\iota)n}$, and balancing the parameters gives $2^{0.60387n+o(n)}$ time complexity in \cref{thm:importance-full-space}.\footnote{Note that since we do not even sample from the discrete Gaussian with width $r$, we can set $r$ smaller than the limit of the original discrete Gaussian sampler \cite{STOC:ADRS15}. We actually choose so, and extend several lemmas for $r<t_0$.}

\paragraph{Sparsity makes space for quantum.}
A careful reader may notice that the above algorithm requires $2^{(2r +\iota)n}$ Gaussian samples, thus its space complexity is also $2^{0.60387n+o(n)}$. We observe that the value $w(X)$ heavily depends on the norm of $X$, and in particular for most of $X$ it is very close to $0$. We can sparsify the Gaussian samples by employing a binary variable $Z$ such that, conditioned on each sample $X$,
choose
$Z\in\{0,1\}$ such that
$\Pr[Z=1\mid X]=\pi(X)$ and
$\Pr[Z=0\mid X]=1-\pi(X)$ for some $\pi(X) \propto w(X)$.
We use the following:
For every function $f$,
\begin{align*}
  \E_{X,Z}\left[
    \frac{Zf(X)}{\pi(X)}
  \right]
  =
  \E_X\left[
    \E\left[
      \left.
      \frac{Zf(X)}{\pi(X)}
      \right|X
    \right]
  \right]
  =
  \E_X[f(X)]
\end{align*}
thus we can compute the statistics with a smaller number of $X$'s for which $Z=1$, which is about $\E[\pi(X)]\cdot 2^{(2r+\iota)n}$. This gives a $2^{0.5n+o(n)}$ space version of the above algorithm, showing \cref{thm:importance-small-space}.
Note that we anyway need to sample $2^{(2r+\iota)n}$ Gaussian samples, thus the time complexity remains unchanged.

This reduced space actually allows us to use the quantum minimum finding algorithm \cite{DH96}. 
The number of effective samples can be reduced to $2^{2rn}$ while ensuring the accuracy of the estimation at best. 
In this case, we choose $b\approx 2rn$ and write
$\F_2^\ell=\F_2^b\times\F_2^{\ell-b}$. For each fixed value of the
last $\ell-b$ coordinates, the fast Walsh-Hadamard transform computes
the corresponding $2^b=2^{2rn+o(n)}$ coset Hessians at once.
Classically this is repeated $2^{\ell-b}$ times, while quantum minimum
finding reduces the number of repetitions to $2^{(\ell-b)/2}$, which
is $2^{\iota n+o(n)}$ for the optimized parameters.
Balancing the parameter gives $2^{0.5411n+o(n)}$ quantum time and $2^{0.5n+o(n)}$ space in \cref{thm:quantum-svp}. The QRAM only stores the BDD preprocessing data and $2^{2rn}$ samples, which is $2^{0.3604n+o(n)}$ in the balanced parameters.

\paragraph{Acknowledgment.} The author is grateful to Yixin Shen for helpful discussions about an early version of this work, and to Jiseung Kim for kindly sharing a
preliminary draft of his ongoing work \cite{Kim26}.
\section{Preliminaries}\label{sec:preliminaries}

The sets of integers, positive integers, and real numbers are denoted by $\Z$, $\N$, and $\R$. 
$\F_2$ denotes the field of two elements $\{0,1\}$.
We write $[m]=\{1,\ldots,m\}$ for $m\in \N$.
The Euclidean, operator, and Frobenius norms are denoted by $\norm{\cdot}$,
$\opnorm{\cdot}$, and $\Frob{\cdot}$, respectively.
The closed Euclidean ball of radius $R$ is $B_2^n(R) =\{x \in \R^n : \norm{x}\le R\}$. 
For an event $E$, its indicator function is $\ind_E$.
All logarithms are natural unless a base is specified.
We write $\dist(t,S):=\inf_{x\in S}\norm{t-x}$ for a set $S$ and a vector $t$.

The dimension $n$ tends to infinity when discussing the asymptotic behavior.
Unless specified otherwise, every $o(1)$ term is uniform over the lattices and the fixed parameters in the stated ranges.
A negligible function, denoted by $\negl(n)$, is smaller than $n^{-c}$ for every $c>0$ for all sufficiently large $n$.
We write $\langle B\rangle$ for the bit length of the input basis $B$ and assume that it is polynomial in $n$.

\subsection{Lattices and decoding}\label{sec:lattices}
We only consider the full-rank lattice defined as follows.\footnote{It is well-known that the full-rank case implies to the general case.}
Let $B=(b_1,...,b_n) \in \R^{n\times n}$ be a matrix with independent columns. A full-rank lattice generated by $B$ is a set $\cL=\cL(B):=B\Z^n$, i.e., integer linear combinations of column vectors of $B$. We define its dual lattice by
\[
  \cL^*
  :=\{y\in{\rm span}(\cL):\ip{y}{x}\in\Z\text{ for every }x\in\cL\}
\] 
which also satisfies $\cL^*=B^{-T}\Z^n$ when $\cL$ is full-rank.

A lattice $\cL'$ is called a sublattice of $\cL$ if $\cL' \subseteq \cL$. If $\cL'$ is a full-rank sublattice, $\cL/\cL'$ is the finite abelian group of cosets $x+\cL'$.
In particular, it holds that
\[
  \cL/2\cL\cong\F_2^n,
  \qquad
  \cL^*/2\cL^*\cong\F_2^n.
\]
For cosets $a+2\cL$ and $x+2\cL^*$, we can write $a=Bu \in \cL$ and $x=B^{-T}z \in \cL^*$
for $u,z \in \F_2^n$. We call $u$ is the parity class.
The following is well-defined
\[
  \chi_a(x):=(-1)^{\ip{x}{a}} =(-1)^{z^Tu}
\]
and independent of the representatives. In the remainder of the paper, we use arbitrary representative of the coset which does not change any results.

The first minimum $\lambda=\lambda_1(\cL)$ is $\lambda_1(\cL):=\min_{x\in\cL\setminus\{0\}}\norm{x}.$
For $R\ge0$, define the counting function for the lattice points inside a ball
$N_{\cL}(R)
  :=\left|(\cL\setminus\{0\})\cap B_2^n(R)\right|.$
We use the following bound \cite[Lemma~3]{EPRINT:PujSte09}, based on the spherical-code bound of \cite{KL78}.

\begin{lemma}\label{thm:lattice-points}
Uniformly for full-rank $n$-dimensional lattices $\cL$ and $x\ge1$, it holds for $\beta=2^{0.4014\ldots}$
\[
  N_{\cL}(x\lambda_1(\cL))
  \le \beta^{n+o(n)}x^n.
\]
\end{lemma}

\begin{definition}[Search-$\SVP$]
Given a basis of a full-rank lattice $\cL$, Search-$\SVP$ asks to find a nonzero $v\in\cL$ with
$\norm{v}=\lambda_1(\cL)$.
\end{definition}

\begin{definition}[$\BDD$]
For $0<\alpha<1/2$, $\alpha$-$\BDD$ takes a lattice
$\cL$ and a target vector $t$ promised to satisfy $\dist(t,\cL)<\alpha\lambda_1(\cL)$ as input,
and returns the unique closest lattice vector to $t$.
\end{definition}

We use the following result for preprocessing $\BDD$ from \cite{SICOMP:ACKS25}.

\begin{theorem}
\label{thm:bdd-preprocessing}
There is a randomized classical preprocessing algorithm for the $n^{-1/3}$-$\BDD$ problem such that:
\begin{itemize}[nosep]
    \item its preprocessing runs in expected time and worst-case space $2^{n/2+o(n)}$ and succeeds with a constant probability. After running, the size of advice is $2^{o(n)}$.
    \item given the successful preprocessing output, any instance of the $n^{-1/3}$-$\BDD$ problem can be deterministically solved in time $2^{o(n)}$ and polynomial additional space.
\end{itemize}
\end{theorem}
\begin{proof}
    Plug $\eps=\exp(-\sqrt{n})$ in \cite[Theorem 56]{SICOMP:ACKS25} with $\alpha=1/2$, which gives $(\phi(\cL)/\lambda_1(\cL))$-BDD oracle in $m \poly(n)$ time and space after the DGS preprocessing $\phi(\cL)= \frac{\sqrt{\ln(1/\eps)/\pi -o(1)}}{2\eta_\eps(\cL^*)}$ and $m=O(n\ln(1/\eps)/\sqrt{\eps}) =2^{o(n)} $. 
    An asymptotic calculation gives $\phi(\cL)/\lambda_1(\cL)  =\Omega(n^{-1/4})$ using $\lambda_1(\cL) \eta_\eps(\cL^*) \le \sqrt{\frac{\ln((1+\eps)/\eps)}{\pi}} + \sqrt{\frac{n}{2\pi}}$. This makes the query in $m \poly(n) = 2^{o(n)}$ time and space. The preprocessing is only the discrete Gaussian samples, which can be sampled in time $2^{n/2+o(n)}$ and space $2^{n/2+o(n)}.$
\end{proof}



We consider a time-truncated variant that returns $\bot$ and aborts if the time is exceeded.
For an input that does not meet the promise of $n^{-1/3}$-$\BDD$, there is no correctness guarantee; we use an output only after exact lattice-membership verification.
Whenever an algorithm is repeated, every trial uses independent preprocessing and fresh randomness, and we return the shortest exactly verified nonzero vector over all trials.
Every $\DGS$ call is time-truncated at its stated budget.
If the call exceeds the budget or returns too few samples, the current scale is aborted.

\subsection{Discrete Gaussians}\label{sec:discrete-gaussian}

For $s>0$, $x\in\R^n$, and a countable set $A$, define
\[
  \rho_s(x):=\exp\!\left(-\pi\frac{\norm{x}^2}{s^2}\right),\qquad\rho_s(A):=\sum_{x\in A}\rho_s(x)
\]
For a full-rank lattice $\cL\subset \R^n$ and a vector $t\in\R^n$,
the discrete Gaussian on a lattice
coset $\cL+t$ is
\[
  D_{\cL+t,s}(x)
  :=\frac{\rho_s(x)}{\rho_s(\cL+t)},
  \qquad x\in\cL+t.
\]
In particular, $D_{\cL,s}$ denotes the central case $t=0.$

\begin{definition}[Smoothing parameter {\cite{FOCS:MicReg04,JACM:Regev09}}]
For $\eps\in(0,1)$, the smoothing parameter is the unique positive number
$\eta_\eps(\cL)$ satisfying
$  \rho_{1/\eta_\eps(\cL)}(\cL^*\setminus\{0\})=\eps.
$
\end{definition}

\begin{lemma}[{\cite[Lemma~6.1]{STOC:ADRS15}}]\label{thm:smoothing-bound}
Let $\beta$ be the constant in \cref{thm:lattice-points}.
For any full-rank lattice $\cL \subset \R^n$, it holds that
\[
  \eta_{1/2}(\cL^*)
  \le
  \left(
    \frac{\beta}{\sqrt{2\pi e}}+o(1)
  \right)
  \frac{\sqrt n}{\lambda_1(\cL)}.
\]
\end{lemma}

\begin{theorem}[{\cite[Theorem~5.11]{STOC:ADRS15}}]\label{thm:dgs-sampling}
Let $\cL\subset\R^n$ be a full-rank lattice, $s>0$, and
$\kappa=\Omega(n)$. There is a classical algorithm $\DGS$ that outputs at
most $M=2^{n/2}$ lattice vectors in time
$2^{n/2+\polylog(\kappa)+o(n)}$ and space $2^{n/2+o(n)}$.
Its output distribution is $\exp(-\Omega(\kappa))$-close to the
following distribution: choose $M'\in\{0,\ldots,M\}$, draw
$X_1,\ldots,X_M$ independently from $D_{\cL,s}$ and independently of
$M'$, and output
$X_1,\ldots,X_{M'}$. If
$s>\sqrt2\eta_{1/2}(\cL)$, then $M'=M$.
\end{theorem}

\begin{lemma}
\label{lem:dgs-tail}
For every $c>0$, there is a constant $C_c>0$ such that, for every lattice
$\cL$, every $s>0$, and $X\sim D_{\cL,s}$,
\[
  \Prob[\norm{X}>C_c s\sqrt n]\le 2^{-2cn}\qquad
\text{and}\qquad
  \E\left[
     \norm{X}^2\ind_{\{\norm{X}>C_c s\sqrt n\}}
  \right]
  \le s^2n\,2^{-cn}.
\]
\end{lemma}

\begin{proof}
The discrete-Gaussian tail inequality {\cite[Lemma~2.4]{STOC:ADRS15}}
gives, for all sufficiently large constants $t$,
\[
\Prob[\norm{X}>ts\sqrt n]
  \le
  \left(\sqrt{2\pi e} \cdot t e^{-\pi t^2}\right)^n.
\]
Choose $t=C_c$ sufficiently large (which only depends on $c$) so that 
\[
\sqrt{2\pi e}\cdot t e^{-\pi t^2} \le e^{\frac{-\pi t^2}{2}} \le 2^{-2c}
\]
holds, which proves the first part. For the second part, the tail integration formula for $Y=\norm{X}$ gives
\[
  \E[Y^2\ind_{\{Y>u\}}]
  =u^2\Prob[Y>u]+
    \int_u^\infty 2z\Prob[Y>z] dz.
\]
Taking $u=C_c s\sqrt n$ and $z=ts \sqrt n$, we have
\begin{align*}
  \E[Y^2\ind_{\{Y>C_c s\sqrt{n}\}}]
  &\le C_c^2 s^2 n 2^{-2cn}+
    s^2 n\int_{C_c}^\infty 2t \cdot  \Prob[\norm{X}>ts\sqrt n] dt\\
    &\le  C_c^2 s^2 n 2^{-2cn}+
    s^2 n\int_{C_c}^\infty 2t \cdot \exp \left(\frac{-n\pi t^2 }{2}\right)dt
    \\
    &= C_c^2 s^2 n 2^{-2cn}+
    s^2 n \cdot \frac{2e^{\frac{-\pi nC_c^2}2 }}{\pi n}\\
    &\le s^2 n \left(C_c^2 + \frac{2}{\pi n}\right)2^{-2cn}
\end{align*}
where the last term is smaller than $s^2 n 2^{-cn}$ for sufficiently large $n$.
\end{proof}

\subsection{Useful inequalities}
We use the following standard probability bounds.

\begin{lemma}[Markov's inequality]\label{lem:Markov}
     For a random variable $Y\ge0$ and $a>0$, $\Prob[Y\ge a]\le{\E[Y]}/a.$
\end{lemma}
\begin{lemma}[Chebyshev's inequality]\label{lem:Cheb}
    If $Y$ has finite variance and $a>0$, 
    $\Prob[|Y-\E[Y]|\ge a]
    \le
    {\operatorname{Var}(Y)}/{a^2}.$
\end{lemma}
\begin{lemma}[Hoeffding's inequality]\label{lem:Hoef}
    If $Y_1,\ldots,Y_N$ are independent and
  $a_i\le Y_i\le b_i$ and $\delta>0$,
  $\Prob\left[
      \left|
        \frac1N\sum_{i=1}^N(Y_i-\E[Y_i])
      \right|
      \ge\delta
    \right]
    \le
    2\exp\left(
      -\frac{2N^2\delta^2}
      {\sum_{i=1}^N(b_i-a_i)^2}
    \right).$
\end{lemma}
\begin{lemma}[Chernoff bound]\label{lem:Chernoff}
If $Y_1,\ldots,Y_N$ are independent Bernoulli random
variables and $0<\delta\le1$, then
\[
  \Prob\left[
    \left|
      \sum_{i=1}^NY_i-\sum_{i=1}^N\E[Y_i]
    \right|
    \ge
    \delta\sum_{i=1}^N\E[Y_i]
  \right]
  \le
  2\exp\left(
    -\frac{\delta^2}{3}\sum_{i=1}^N\E[Y_i]
  \right).
\]
\end{lemma}

\begin{lemma}\label{lem:majority}
If $E_1,\ldots,E_m$ are independent events with
$\Prob[E_i]\le p<1/4$, then
\[
  \Prob\left[
    \sum_{i=1}^m\ind_{E_i}\ge\frac m2
  \right]
  \le
  2^m p^{m/2}.
\]
\end{lemma}

\section{Recovering Shortest Vector from Hessian}\label{sec:hessian-recovery}


We assume that an approximate shortest-vector length $d$ satisfying
\begin{align}
    \lambda_1(\cL) \le d \le (1+1/n) \lambda_1(\cL)
\end{align}
is known. This can be assumed by applying LLL \cite{LLL82} and
guessing $d$ from
$\{(1+1/n)^{-j}\norm{x}\}_{j=0,\ldots,n^2}$, where $x$ is the
resulting nonzero vector. We formalize this in the final algorithm.

It will be convenient to define the following parameters:
\begin{align}\label{eqn: def_xi_t0}
  \xi_t=\xi_t(d)
  :=
  \sqrt{\frac{4nt\ln2}{\pi d^2}}, \qquad
  t_0:= \frac{\beta^2}{4e \ln 2} = 0.23147\ldots
\end{align}
These choices give the Gaussian-mass estimates in
\cref{cor:working-scale}. It will be convenient to note that for most of our interested case the following holds for most choices of $t$:
\begin{align}\label{eqn: easy_xi}
    \xi_t^2 = \Theta(n/d^2) = \Theta(n/\lambda^2).
\end{align}

\subsection{Scale and Gaussian mass bounds}\label{sec:working-scales}

\begin{lemma}\label{lem:gaussian-moments}
Let $\cL$ be an $n$-dimensional full-rank lattice and let
$\lambda=\lambda_1(\cL)$.  Let $s=\Theta(\lambda/\sqrt n)$ and let $k\ge 0$ be a constant.
Then
\[
  \sum_{x\in\cL\setminus\{0\}}
  \left(1+\frac{\norm{x}}{\lambda}\right)^k
  \rho_s(x)
  \le
  2^{o(n)}
  \left(
    \frac{\beta^2ns^2}
         {2\pi e\lambda^2}
  \right)^{n/2}.
\]
\end{lemma}

\begin{proof}
The case $k=0$ is essentially
\cite[Lemma~4.2]{STOC:ADRS15}.  The same shell argument gives
\[
  \sum_{x\in\cL\setminus\{0\}}
  \left(1+\frac{\norm{x}}{\lambda}\right)^k
  \rho_s(x)
  \le
  2^k\beta^{n+o(n)}
  \max_{y\ge0}
  y^{n+k}e^{-\pi\lambda^2y^2/s^2}
  \le
  \beta^{n+o(n)}
  \max_{y>0}
  y^{n+k}e^{-\pi\lambda^2y^2/s^2}.
\]
Here we used $\norm{x}\ge\lambda$ and hence
\[
  \left(1+\frac{\norm{x}}{\lambda}\right)^k
  \le
  2^k\left(\frac{\norm{x}}{\lambda}\right)^k.
\]
The maximum is attained at
$  y^2=\frac{(n+k)s^2}{2\pi\lambda^2},$
and its value is
$  \left(
    \frac{(n+k)s^2}
         {2\pi e\lambda^2}
  \right)^{(n+k)/2}.
$
Since $s=\Theta(\lambda/\sqrt n)$ and $k$ is fixed, replacing
$n+k$ by $n$ changes the expression by only $2^{o(n)}$.
The fixed factor $2^k$ and the subexponential shell-summation factor are absorbed into the same term.
\end{proof}




\begin{corollary}\label{cor:working-scale}
Let $\lambda=\lambda_1(\cL)$ and
$\lambda \le d\le (1+1/n)\lambda$.
For every fixed $a>0$ and $k\ge0$,
\[
  \sum_{x\in\cL\setminus\{0\}}
  \left(1+\frac{\norm{x}}{\lambda_1(\cL)}\right)^k
  \rho_{a/\xi_t}(x)
  \le
  2^{
    \left(
      \frac12\log_2\frac{a^2t_0}{2t}
      +o(1)
    \right)n
  }.
\]
Consequently, for every fixed $a>0$ and $k>0$,
\begin{align}\label{eqn: xkrho_bound}
    \sum_{x\in\cL}\norm{x}^k\rho_{a/\xi_t}(x)
\le
\lambda^k
2^{(\frac12\log_2\frac{a^2t_0}{2t}+o(1))n}
\end{align}
In particular, the following hold:
\begin{itemize}[nosep]
    \item For a (possibly negative) constant $c$, it holds that
$\rho_{\sqrt{2^{1+c}}/\xi_t}(\cL)
  \le
  1+2^{(c/2-\frac{1}{2}\log_2\frac{t}{t_0}+o(1))n}.$
    \item If $t>t_0/2$, then
  $\rho_{1/\xi_t}(\cL)
  =
  1+2^{-\Omega(n)}.$
    \item If $t>t_0$, then
$  \xi_t>\sqrt2\,\eta_{1/2}(\cL^*),
$
and \cref{thm:dgs-sampling} applies to $\cL^*$ at parameter $\xi_t$.
\end{itemize}
\end{corollary}

\begin{proof}
For $s=a/\xi_t$, we have
$  \frac{\beta^2ns^2}{2\pi e\lambda^2}
  =
  \frac{a^2t_0}{2t}\frac{d^2}{\lambda^2}.$
The first claim follows from \cref{lem:gaussian-moments}, since
$d/\lambda\le1+1/n$.
The second claim follows from the first because its term at $x=0$
vanishes and
$\norm{x}^k\le
\lambda^k(1+\norm{x}/\lambda)^k$ for $x\ne0$.
The first bound follows by taking
$a=\sqrt{2^{c+1}}$ for given $c$ and $k=0$. Plugging $t>t_0/2$
gives the second bound.
If $t>t_0$, then the bound gives
\[
  \rho_{\sqrt2/\xi_t}(\cL\setminus\{0\})
  <\frac12
\]
for all sufficiently large $n$.
The definition and monotonicity of the smoothing parameter give the last claim.
\end{proof}

\subsection{The periodic Gaussian and the shortest parity class}
\label{sec:periodic-hessian}
\label{sec:shortest-class-hessian}

For $s>0$, define the periodic Gaussian function
\cite{AR05,CCC:DRS14} by
\begin{align}\label{eqn: periodic Gaussian}
    F_s(z)
  :=
  \frac{\rho_s(\cL+z)}
       {\rho_s(\cL)},
  \qquad z\in\R^n.
\end{align}
The following is the Poisson representation of $F_s$.
\begin{align}\label{lem:poisson}
  F_s(z)
  =
  \E_{X\sim D_{\cL^*,1/s}}
  \left[
    e^{2\pi i\ip{X}{z}}
  \right].
\end{align}

For $u\in\F_2^n$, define the Hessian
\[
  \cG_{t,d}(u)
  :=
  \nabla^2 F_{1/\xi_t}(Bu/2).
\]
We also define
\begin{align}
    \label{eqn:cAdef}
  \cA_{t,d}(u)
  :=
  \cG_{t,d}(u)
    +
    2\pi\xi_t^2
    F_{1/\xi_t}(Bu/2)\Id.
\end{align}
For fixed $u$, $\cA_{t,d}(u)$ and $\cG_{t,d}(u)$ differ by
$2\pi\xi_t^2F_{1/\xi_t}(Bu/2)\Id$, a constant multiple of $\Id$. Hence
the two matrices have the same eigenvectors and eigengaps.
We will use the following alternative formula.

\begin{lemma}\label{lem:hessian-formulas}
For every $u\in\F_2^n$,
\begin{align}\label{eqn: cGexpectation}
  \cG_{t,d}(u)
  =
  -4\pi^2
  \E_{X\sim D_{\cL^*,\xi_t}}
  \left[
    XX^T
    (-1)^{u^T(B^TX\bmod2)}
  \right].
\end{align}
Moreover, the following equality holds thus $\cA_{t,d}(u)$ is positive semidefinite:
\begin{align}\label{eqn: cApositive}
  \cA_{t,d}(u)
  =
  \frac{\pi^2\xi_t^4}{\rho_{1/\xi_t}(\cL)}
  \sum_{w\in Bu+2\cL}
  ww^T
  e^{-\pi\xi_t^2\norm{w}^2/4}.
\end{align}
\end{lemma}

\begin{proof}
Differentiating the expression in
\cref{lem:poisson} gives
\[
  \nabla^2F_s(z)
  =
  -4\pi^2
  \E_{X\sim D_{\cL^*,1/s}}
  \left[
    XX^T e^{2\pi i\ip{X}{z}}
  \right].
\]
Set $s=1/\xi_t$ and $z=Bu/2$.
Since $B^TX\in\Z^n$,
\[
  e^{2\pi i\ip{X}{Bu/2}}
  =
  e^{\pi i u^TB^TX}
  =
  (-1)^{u^T(B^TX\bmod2)}.
\]
This proves the first identity.

Differentiating the expression in \cref{eqn: periodic Gaussian} and
using $\cL=-\cL$ gives
\[
  \nabla^2F_s(z)
  =
  \frac1{\rho_s(\cL)}
  \sum_{y\in\cL}
  \left(
    \frac{4\pi^2}{s^4}(y-z)(y-z)^T
    -
    \frac{2\pi}{s^2}\Id
  \right)                                    
  \cdot
  e^{-\pi\norm{y-z}^2/s^2}.
\]
It follows that
\[
    \nabla^2F_s(z)+\frac{2\pi}{s^2}F_s(z)\Id
  =
  \frac{4\pi^2}{s^4\rho_s(\cL)}
  \sum_{y\in\cL}
  (y-z)(y-z)^T
  e^{-\pi\norm{y-z}^2/s^2}.
\]
Again set $s=1/\xi_t$, $z=Bu/2$, and reindex by
$w=2y-Bu$. 
The second identity follows from
$-Bu+2\cL=Bu+2\cL$.
\end{proof}

The following lemma shows that, if a shortest vector $v$ is included in the coset $Bu+2\cL$, the matrix $\cA_{t,d}(u)$ is approximately proportional to the matrix $vv^T$.

\begin{lemma}\label{lem:shortest-class-hessian}
Let $t>t_0$ and
$  \lambda=\lambda_1(\cL)
  \le d\le
  (1+1/n)\lambda.$
Let $v\in\cL$ satisfy $\norm{v}=\lambda$ and
$u\in\F_2^n$ satisfy
$  v\in Bu+2\cL.$
For the normalized vector $\widehat v=v/\lambda$,
it holds that
\begin{align}\label{eqn: mudef}
  \cA_{t,d}(u)
  =
  \mu_{t,d}\widehat v\widehat v^T
  +
  \cR_{t,d}(u)\qquad \text{ for }\quad   \mu_{t,d}
  :=
  \frac{
    2\pi^2\xi_t^4\lambda^2
  }{
    \rho_{1/\xi_t}(\cL)
  }
  e^{-\pi\xi_t^2\lambda^2/4}=
  \xi_t^4\lambda^2
  2^{-tn+O(1)},
\end{align}
where $\cR_{t,d}(u)$ is positive semidefinite such that
$  \norm{\cR_{t,d}(u)}_{\mathrm{op}}
  \le
  \mu_{t,d}\cdot
  2^{
    \left(
      -\frac12\log_2\frac{t}{t_0}
      +o(1)
    \right)n
  }.$
\end{lemma}

\begin{proof}
Since
$  Bu+2\cL=v+2\cL,$
the terms $w=v$ and $w=-v$ in
\cref{eqn: cApositive} contribute
\[
  \frac{
    2\pi^2\xi_t^4
  }{
    \rho_{1/\xi_t}(\cL)
  }
  vv^T
  e^{-\pi\xi_t^2\lambda^2/4}
  =
  \mu_{t,d}\widehat v\widehat v^T.
\]
Let $\cR_{t,d}(u)$ be the sum of the remaining terms.
It is positive semidefinite.

Write every element of $v+2\cL$ as
$w=-v+2y$ for $y\in \cL$.
The two removed terms correspond to $y=0$ and $y=v$.
Since $\cR_{t,d}(u)$ is positive semidefinite, we have
$  \norm{\cR_{t,d}(u)}_{\mathrm{op}}
  \le
  \operatorname{Tr}\left(\cR_{t,d}(u)\right).$ This gives
\[
\begin{split}
  \frac{
    \norm{\cR_{t,d}(u)}_{\mathrm{op}}
  }{
    \mu_{t,d}
  }
  &\le
  2
  \sum_{y\in\cL\setminus\{0,v\}}
  \frac{\norm{y-v/2}^2}{\lambda^2}
  e^{
    -\pi\xi_t^2
    \left(
      \norm{y-v/2}^2-\lambda^2/4
    \right)
  }.
\end{split}
\]
For $y\notin\{0,v\}$, both $y$ and $y-v$ are nonzero
lattice vectors. Hence
\[
  \norm{y-v/2}^2-\frac{\lambda^2}{4}
  =
  \frac{
    \norm{y}^2+\norm{y-v}^2-\lambda^2
  }{2}
  \ge
  \frac{\norm{y}^2}{2}.
\]
Moreover,
\[ \frac{\norm{y-v/2}}{\lambda}
  \le
    \frac{\norm{y}}{\lambda}+\frac{\norm{v}}{2\lambda}
  \le
  \frac{\norm{y}}{\lambda}+1
\]
Combining these two bounds and \cref{cor:working-scale} with
$a=\sqrt2$ and $k=2$ gives
\[
  \frac{
    \norm{\cR_{t,d}(u)}_{\mathrm{op}}
  }{
    \mu_{t,d}
  }\le
  2  \sum_{y\in\cL\setminus\{0\}}
  \left(
    1+\frac{\norm{y}}{\lambda}
  \right)^2
  \rho_{\sqrt2/\xi_t}(y)
  \le
  2^{
    \left(
      -\frac12\log_2\frac{t}{t_0}
      +o(1)
    \right)n
  }.
\]

Finally, the asymptotic of $\mu_{t,d}$ follows from the choice of $\xi_t$ in \cref{eqn: def_xi_t0} and $\lambda \le d \le (1+1/n)\lambda$, as well as $\rho_{1/\xi_t}(\cL)
  =
  1+2^{-\Omega(n)}$
by \cref{cor:working-scale}.
\end{proof}

\subsection{Estimating the Hessians and recovering a shortest vector}
\label{sec:hessian-estimation}
\label{sec:recovering-shortest-vector}
\label{sec:direct-hessian}

By \cref{lem:shortest-class-hessian}, the exact Hessian at the parity
class of a shortest vector $v$ contains some information about $v$.
In this section, we show how to actually recover $v$ using this property. As discussed in the overview, this can be done by the three steps; 1) compute an estimation of Hessian, 2) compute a unit vector close to the direction of $v$, and 3) recover the vector $v$ using (preprocessing) $\BDD$ oracle from \cref{thm:bdd-preprocessing}. 

Fix $0<t<1/4$ and let
$N_t=\lceil n^5\log n\cdot 2^{2tn}\rceil$.\footnote{Throughout this paper, we do not try to optimize the polynomial factors.}
For independent $X_1,\ldots,X_{N_t}\sim D_{\cL^*,\xi_t}$ and
$u\in\F_2^n$, define
\begin{align}
\label{eqn: Gesti} 
  \widehat{\cG}_{t,d}(u)
  :=
  -\frac{4\pi^2}{N_t}
  \sum_{i=1}^{N_t}
  X_iX_i^T(-1)^{u^T(B^TX_i\bmod2)}
  \ind_{\{\norm{X_i}\le C_1\xi_t\sqrt n\}},   
\end{align}
where $C_1$ is the constant in \cref{lem:dgs-tail} for $c=1$. The parameters will be chosen so that we can use the $\DGS$ algorithm in \cref{thm:dgs-sampling}.

The following lemma shows that $\widehat \cG$ approximate $\cG$ accurately.

\begin{lemma}\label{lem:hessian-estimation}
For every fixed $0<t<1/4$, with probability
$1-2^{-\Omega(n\log n)}$, it holds for all $u\in \F_2^n$ that
\[
  \max_{u\in\F_2^n}
  \opnorm{
    \widehat{\cG}_{t,d}(u)-\cG_{t,d}(u)
  }
  \le
  2\xi_t^2 2^{-tn}.
\]
\end{lemma}

\begin{proof}

Let $X\sim D_{\cL^*,\xi_t}$. Noting that $\norm{X}>C_1 \xi_t \sqrt{n}$ is truncated in $\widehat{\cG}$,
\[
  \E[\widehat{\cG}_{t,d}(u)]-\cG_{t,d}(u)
  =
  4\pi^2
  \E\left[
    XX^T
    (-1)^{u^T(B^TX\bmod2)}
    \ind_{\{\norm{X}>C_1\xi_t\sqrt n\}}
  \right].
\]
Using
$\opnorm{\E[M]}\le\E[\opnorm{M}]$ and
$\opnorm{XX^T}=\norm{X}^2$, Lemma~\ref{lem:dgs-tail} with $c=1$
gives for sufficiently large $n$
\begin{align}\label{eqn: lem35_opbound}
  \opnorm{
    \E[\widehat{\cG}_{t,d}(u)]-\cG_{t,d}(u)
  }
  \le
  4\pi^2
  \E\left[
    \norm{X}^2
    \ind_{\{\norm{X}>C_1\xi_t\sqrt n\}}
  \right]
  \le
  4\pi^2\xi_t^2n2^{-n} \le \xi_t^22^{-tn}.
\end{align}

We next bound the probability that the estimation $\widehat\cG$ is far from $\cG$. Fix $u\in\F_2^n$ and
$a,b\in\{1,\ldots,n\}$. The absolute value of the $(a,b)$-entry
of each summand in \cref{eqn: Gesti} is at most
\[
  4\pi^2
  |(X_i)_a(X_i)_b|
  \ind_{\{\norm{X_i}\le C_1\xi_t\sqrt n\}}
  \le
  4\pi^2C_1^2\xi_t^2n.
\]
Therefore, Hoeffding's inequality gives, using $N_t=\lceil n^5\log n\cdot2^{2tn}\rceil$,
\[
  \Prob\left[
    \left|
      \left(
        \widehat{\cG}_{t,d}(u)
        -
        \E[\widehat{\cG}_{t,d}(u)]
      \right)_{a,b}
    \right|
    >
    \frac{\xi_t^2}{n}2^{-tn}
  \right]
  \le
  2\exp\left(
    -\Omega\left(
      \frac{N_t}{n^4}2^{-2tn}
    \right)
  \right)
  =
  2e^{-\Omega(n\log n)}.
\]
By the union bound over all $2^n n^2$ choices of
$u\in\F_2^n$ and $a,b\in\{1,\ldots,n\}$,
with probability $1-2^{-\Omega(n\log n)}$, every entry of every difference $\widehat{\cG}_{t,d}(u)
        -
        \E[\widehat{\cG}_{t,d}(u)]$ has absolute value at most
$\xi_t^2 2^{-tn}/n$. Since
$  \opnorm{M}
  \le
  n\max_{a,b}|M_{a,b}|,$
the operator norm of every difference is at most $\xi_t^2 2^{-tn}$.
Together with \cref{eqn: lem35_opbound}, this concludes the proof.
\end{proof}

Given that $\widehat{\cG}_{t,d}(u) \approx {\cG}_{t,d}(u)  \approx \mu_{t,d} \widehat{v} \widehat{v}^T$ due to \cref{lem:shortest-class-hessian,lem:hessian-estimation}, the eigenvector of $\widehat{\cG}_{t,d}(u)$ must be close to the direction of $v$. The following lemma formalize this intuition.

\begin{lemma}\label{lem:hessian-recovery}
Let $t_0<t<1/4$ be fixed and let
$\lambda=\lambda_1(\cL)\le d\le(1+1/n)\lambda$.
Let $v\in\cL$ satisfy $\norm{v}=\lambda$, and let
$u\in\F_2^n$ satisfy $v\in Bu+2\cL$.
Suppose that the conclusion of \cref{lem:hessian-estimation} holds.
If $q$ is a unit eigenvector corresponding to the largest eigenvalue
of $\widehat{\cG}_{t,d}(u)$, then
\[
  \min_{\sigma\in\{-1,1\}}
  \norm{q-\sigma v/\lambda}
  =
  O(n^{-1/2}).
\]
Consequently, the $n^{-1/3}$-$\BDD$ query at $dq$ returns one of
$\pm v$ for all sufficiently large $n$.
\end{lemma}

\begin{proof}
Let $\widehat v=v/\lambda$. By
\cref{lem:shortest-class-hessian} and the definition of
$\cA_{t,d}(u)$ in \cref{eqn:cAdef},
\begin{align}
    \label{eqn: cGapporox}
  \widehat{\cG}_{t,d}(u)
  =
  -2\pi\xi_t^2F_{1/\xi_t}(Bu/2)\Id
  +
  \mu_{t,d}\widehat v\widehat v^T
  +
  E,
\end{align}
where
$E=\cR_{t,d}(u)+\widehat{\cG}_{t,d}(u)-\cG_{t,d}(u)$ and $\opnorm{\cR_{t,d}(u)}/\mu_{t,d}=2^{-\Omega(n)}$.

Since $\lambda\le d\le(1+1/n)\lambda$,
$\xi_t^2\lambda^2=\Theta(n)$.
Also,
$\mu_{t,d}=\Theta(\xi_t^4\lambda^2 2^{-tn})$
by \cref{lem:shortest-class-hessian}.
\cref{lem:hessian-estimation} gives
\[
  \frac{
    \opnorm{
      \widehat{\cG}_{t,d}(u)-\cG_{t,d}(u)
    }
  }{\mu_{t,d}}
  =
  O\left(
    \frac1{\xi_t^2\lambda^2}
  \right)
  =
  O(n^{-1}).
\]
Together with
$\opnorm{\cR_{t,d}(u)}/\mu_{t,d}=2^{-\Omega(n)}$,
this gives $\opnorm{E}=O(\mu_{t,d} n^{-1})$.

Since $q$ is a eigenvector corresponding to the largest
eigenvalue of $\widehat{\cG}_{t,d}(u)$, we have
$q^T\widehat{\cG}_{t,d}(u)q
\ge
\widehat v^T\widehat{\cG}_{t,d}(u)\widehat v$.
Substituting the expression for $\widehat{\cG}_{t,d}(u)$ using \cref{eqn: cGapporox}, the
multiple of $\Id$ has the same value on both sides because
$\norm{q}=\norm{\widehat v}=1$.
The remaining inequality is
$\mu_{t,d}\ip{q}{\widehat v}^2+q^TEq
\ge
\mu_{t,d}+\widehat v^TE\widehat v$.
Since $|q^TEq|\le\opnorm{E}$ and
$|\widehat v^TE\widehat v|\le\opnorm{E}$, it follows that
\[
  \mu_{t,d}
  \left(
    1-\ip{q}{\widehat v}^2
  \right)
  \le q^T E q - \widehat v^T E \widehat v\le
  2\opnorm{E}= O(\mu_{t,d}n^{-1}).
\]

Choose $\sigma\in\{-1,1\}$ such that
$\ip{q}{\sigma\widehat v}\ge0$. It follows that
$\norm{q-\sigma\widehat v}^2
\le2(1-\ip{q}{\widehat v}^2)=O(n^{-1})$.

Finally,
$\norm{dq-\sigma v}
\le d\norm{q-\sigma\widehat v}+d-\lambda
=O(\lambda n^{-1/2})<\lambda n^{-1/3}$
for all sufficiently large $n$.
The last claim follows from \cref{thm:bdd-preprocessing}.
\end{proof}

Combining the estimate and the recovery argument gives the direct
algorithm.

\begin{algorithmblock}{Direct Hessian $\SVP$}\label{alg:direct-exhaustive}
  \item Apply LLL reduction to obtain a nonzero $x\in\cL$ and use the scales
        $d_j=(1+1/n)^{-j}\norm{x}$ for $0\le j\le n^2$.
        Construct the preprocessing $n^{-1/3}$-$\BDD$ data.
  \item Fix $t_0<t<1/4$.
        At every scale $d_j$, call $\DGS$ from \cref{thm:dgs-sampling} with
        $\kappa=n^2$ and keep $N_t$ samples from
        $D_{\cL^*,\xi_t(d_j)}$.
        Abort the scale if fewer than $N_t$ samples are returned.
  \item For every $u\in\F_2^n$, evaluate
        $\widehat{\cG}_{t,d_j}(u)$ directly and compute a unit eigenvector
        $q_u$ corresponding to its largest eigenvalue.
        Query $\BDD$ at $d_jq_u$ and retain the output only after exact
        lattice-membership verification.
  \item Return the shortest verified nonzero vector over all scales and
        all $u$.
\end{algorithmblock}

\begin{theorem}\label{thm:direct-hessian-svp}
\cref{alg:direct-exhaustive} solves
Search-$\SVP$ with constant probability in time $2^{1.46295n+o(n)}$
and space
$  2^{n/2+o(n)}.$
\end{theorem}
\begin{proof}
The correctness follows from the previous lemmas. \cref{lem:hessian-estimation} shows that $\widehat{\cG}$ estimate $\cG$ for all inputs with overwhelming probability, and in that case \cref{lem:hessian-recovery} and \cref{thm:bdd-preprocessing} show the correctness of the final outcome.

We explain the time and space complexity below.
There are polynomially many choices of $d_j$ because of the property of LLL-reduced vectors. For each $d_j$, the
algorithm considers all $2^n$ values of $u \in \F_2^n$, and evaluating
$\widehat{\cG}_{t,d_j}(u)$ takes $N_t\poly(n)$ time for $N_t= \poly(n)2^{2tn}$. Hence the total
time is $2^{(1+2t)n+o(n)}.$
Because of \cref{lem:shortest-class-hessian}, 
we have $1+2t> 1+2t_0=1.4629\ldots,$ thus choosing $t\to t_0$ gives the desired time complexity.
All remaining steps take less time and do not affect the overall time complexity. The space complexity is due to $\DGS$ in \cref{thm:dgs-sampling}.
\end{proof}




\section{Batch Hessian Estimation}
\label{sec:walsh}

The factor $2^n$ in the $2^{1.4629n+o(n)}$ running time of the
previous section comes from evaluating
$\widehat{\cG}_{t,d}(u)$ separately for every $u\in\F_2^n$.
The definition of $\widehat{\cG}_{t,d}(u)$ uses the same samples for
all $u$, and \cref{lem:hessian-estimation} proves that all these
estimators are simultaneously accurate. 
This section optimizes the algorithm by batching the Hessian estimations. The main idea is to use the Walsh-Hadamard transform in the following lemma.

\begin{lemma}\label{lem:matrix-wht}
Let $A:\F_2^\ell\to\R^{n\times n}$. Given $A(x)$ for every
$x\in\F_2^\ell$, the matrices
\[
  \sum_{x\in\F_2^\ell}(-1)^{\theta\cdot x}A(x),
  \qquad \theta\in\F_2^\ell,
\]
can be computed in time $2^\ell\poly(n)$ and space
$2^\ell\poly(n)$.
\end{lemma}

\begin{proof}
Store the matrices $A(x)$ in an array indexed by $\F_2^\ell$.
For every $j\in\{1,\ldots,\ell\}$ and every $x\in\F_2^\ell$
with $x_j=0$, replace
\[
  \bigl(A(x),A(x+e_j)\bigr)
  \quad\text{by}\quad
  \bigl(A(x)+A(x+e_j),A(x)-A(x+e_j)\bigr),
\]
where $e_j$ is the $j$th standard basis vector.
After processing all $\ell$ coordinates, the entry indexed by
$\theta$ equals
$\sum_x(-1)^{\theta\cdot x}A(x)$.
There are $\ell2^{\ell-1}$ replacements, and the array contains
$2^\ell$ matrices.
\end{proof}

\begin{algorithmblock}{Walsh-Hadamard Hessian $\SVP$}
\label{alg:full-scan}
  \item Construct the preprocessing $\BDD$ data and the scale grid,
        choose $t$, and obtain the Discrete Gaussian samples at every scale $d_j$ as in
        \cref{alg:direct-exhaustive}. Set $\ell=\lfloor n/2\rfloor$.

  \item At each scale $d_j$ that was not aborted, define, for every
        sample $X_i$,
        \[
          k_i:=B^TX_i\bmod2,
          \qquad
          W_i:=
          -\frac{4\pi^2}{N_t}X_iX_i^T
          \ind_{\{\norm{X_i}\le C_1\xi_t(d_j)\sqrt n\}}.
        \]
        Write
        $k_i=(k_i',k_i'')\in
        \F_2^\ell\times\F_2^{n-\ell}$.

  \item For every $\theta''\in\F_2^{n-\ell}$, construct the array
        $A_{\theta''}:\F_2^\ell\to\R^{n\times n}$ given by
        \[
          A_{\theta''}(x)
          :=
          \sum_{\substack{1\le i\le N_t\\k_i'=x}}
          (-1)^{\theta''\cdot k_i''}W_i.
        \]
        Apply \cref{lem:matrix-wht} to $A_{\theta''}$.
        The matrix indexed by $\theta'\in\F_2^\ell$ is
        $\widehat{\cG}_{t,d_j}(\theta',\theta'')$.

        For each such matrix, compute a unit eigenvector $q$
        corresponding to its largest eigenvalue, query $\BDD$ at
        $d_jq$, and retain the output if it is a nonzero lattice
        vector.

  \item Return the shortest retained vector over all scales.
\end{algorithmblock}

\begin{theorem}\label{thm:full-scan}
\Cref{alg:full-scan} solves Search-$\SVP$ in time 
$  2^{n+o(n)}$ and 
$  2^{n/2+o(n)}
$ space with constant success
probability.
\end{theorem}

\begin{proof}
For $u=(\theta',\theta'')$, the matrix produced by the transform is
\[
  \sum_{x\in\F_2^\ell}
  (-1)^{\theta'\cdot x}A_{\theta''}(x)
  =
  \sum_{i=1}^{N_t}
  (-1)^{
    \theta'\cdot k_i'
    +
    \theta''\cdot k_i''
  }W_i=
  \widehat{\cG}_{t,d_j}(u).
\]
Thus the algorithm evaluates exactly the estimators defined in
\cref{eqn: Gesti}. The correctness then follows the same argument as in \cref{thm:direct-hessian-svp}.

Now we discuss the space and time complexity. In the algorithm, only one array $A_{\theta''}$ is stored at a time. The array contains
$2^\ell$ matrices, while the number of Discrete Gaussian samples is
$N_t<2^{n/2}$. Together with the preprocessing $\BDD$ data, this
gives space
$2^{n/2+o(n)}$.

For each fixed $\theta''$, constructing the array
$A_{\theta''}$ takes $N_t\poly(n)$ time and the following Hadamard-Walsh transform take
$2^{\ell+o(n)}$ time
by \cref{lem:matrix-wht}. Since there are
$2^{n-\ell}$ choices of $\theta''$, the total time for each $d_j$ is
$2^{n-\ell}(N_t+2^\ell)2^{o(n)}$.
Here $\ell=\lfloor n/2\rfloor$ and
$N_t=2^{2tn+o(n)}<2^{n/2}$ because $t<1/4$. Therefore, the total
time over the polynomially many scales is
$2^{n+o(n)}$.
\end{proof}

\section{Random Sublattice Coset Hessian}
\label{sec:random-row}
The algorithm in \cref{sec:walsh} computes $\widehat{\cG}(u)$ for all $u\in \F_2^n$, resulting in the complexity $2^{n+o(n)}$. 
This section shows that the same strategy works by estimating the Hessians using vectors in a sublattice cosets.



As always, we mainly focus on the case $\lambda \le d \le (1+1/n)\lambda$ for $\lambda=\lambda_1(\cL)$. We write
$g(t):=\frac12\log_2(t/t_0)$ for $t>0$. For our main interested parameter $d$,
\cref{cor:working-scale} gives, for every fixed $c$ and positive integer $k$,
\begin{align}\label{eqn:simplerho_g}
  \sum_{x\in\cL}\norm{x}^k\rho_{\sqrt{2^{1+c}}/\xi_t}(x)
\le
\lambda^k
2^{( c/2 -g(t)+o(1))n},\qquad\text{and}\qquad
  \rho_{\sqrt{2^{1+c}}/\xi_t}(\cL)
  \le
  1+2^{(c/2-g(t)+o(1))n}.
\end{align}

The following term from \cref{lem:shortest-class-hessian} is convenient for us.
\begin{align}\label{eqn: mutdconv}
    \mu_{t,d} = \xi_t^4 \lambda^2 2^{-tn+O(1)}  =\xi_t^4 d^2 2^{-tn+O(1)} \qquad\text{for}\quad \lambda\le d \le (1+1/n) \lambda.
\end{align}

\subsection{Gaussian and Hessian sums}
\label{sec:spectral-energy}
The following periodic Gaussian and Hessian mass bounds will be used in this section.
\begin{lemma}\label{lem:spectral-identities}
For every $s>0$,
\[
  \sum_{u\in\F_2^n}F_s(Bu/2)
  =
  \frac{\rho_{2s}(\cL)}{\rho_s(\cL)}.
\]
If $\lambda \le d\le (1+1/n) \lambda$, it holds that
$  \sum_{u\ne0}F_{1/\xi_t}(Bu/2)
  \le
  2^{(1/2-g(t)+o(1))n}.$
\end{lemma}

\begin{proof}
The cosets $Bu/2+\cL$ for $u\in\F_2^n$ partition
$\frac12\cL$. Hence,
\[
  \sum_{u\in\F_2^n}F_s(Bu/2)
  =
  \sum_{u\in\F_2^n} \frac{\rho_s(Bu/2 + \cL)}{\rho_s(\cL)}
  =
  \frac{\rho_s(\frac12\cL)}{\rho_s(\cL)}
  =
  \frac{\rho_{2s}(\cL)}{\rho_s(\cL)}.
\]
The final statement follows from \cref{eqn:simplerho_g} with $c=1$.
\end{proof}

\begin{lemma}\label{lem:hessian-square-mass}
Let $\lambda \le d\le (1+1/n)\lambda$ and $t>t_0/2$.
Then
\[
  \sum_{u\in\F_2^n}\Frob{\cA_{t,d}(u)}^2
  \le
  \mu_{t,d}^2
  2^{(2t-\min\{g(t),2g(t)\}+o(1))n}.
\]
\end{lemma}

\begin{proof}
We use $\Frob{X}^2 = \operatorname{Tr}(X^T X)$ for the real matrix $X$. Applying this to \cref{eqn: cApositive}, the left-hand side becomes
\[
\left(\frac{\pi^2\xi_t^4}{\rho_{1/\xi_t}(\cL)}\right)^2
\sum_{u \in \F_2^n} 
\sum_{z,z'\in Bu+2\cL} 
\operatorname{Tr}(zz^Tz'z'^T) e^{-\pi\xi_t^2(\norm{z}^2+\norm{z'}^2)/4}.
\]
As $u$ varies, the map from $(z,z')\in(Bu+2\cL)^2$ defined by $(z,z') \mapsto (p=\frac{z+z'}2,q=\frac{z-z'}2)$ is a bijection to $\cL^2$, and it satisfies 
\[
  \operatorname{Tr}(zz^Tz'z'^T)
  =
  (z^Tz')^2
  =
  (\norm{p}^2-\norm{q}^2)^2,\qquad \norm{z}^2+\norm{z'}^2 = 2(\norm{p}^2+\norm{q}^2).
\]
Therefore,
\[
  \sum_{u\in\F_2^n}\Frob{\cA_{t,d}(u)}^2
  =
  \frac{\pi^4\xi_t^8}{\rho_{1/\xi_t}(\cL)^2}
  \sum_{p,q\in\cL}
  (\norm{p}^2-\norm{q}^2)^2
  e^{-\pi\xi_t^2(\norm{p}^2+\norm{q}^2)/2}.
\]
Let
$w(x)=e^{-\pi\xi_t^2\norm{x}^2/2}
=\rho_{\sqrt2/\xi_t}(x)$.
Using
$(\norm{p}^2-\norm{q}^2)^2
\le2(\norm{p}^4+\norm{q}^4)$,
the sum over $p,q\in\cL$ is at most
\[
  4
  \left(
    \sum_{p\in\cL}
    \norm{p}^4w(p)
  \right)
  \left(
    \sum_{q\in\cL}
    w(q)
  \right).
\]
Indeed, the Gaussian weight factors as $w(p)w(q)$, and the
two terms obtained from $\norm{p}^4+\norm{q}^4$ are equal.
Applying the first inequality with $k=4$ and $c=0$ of \cref{eqn:simplerho_g} to the first term and the second inequality with $c=0$ to the second term,
we have an upper bound
\[
  \lambda^4
  2^{(-g(t)+o(1))n} \cdot (1+2^{(-g(t)+o(1))n}) = \lambda^4
  2^{(-\min\{g(t),2g(t)\}+o(1))n}.
\]
The resulting exponent is $-g(t)$ if $g(t)\ge0$ and
$-2g(t)$ otherwise. The fixed constants and
$\rho_{1/\xi_t}(\cL)^{-2}\le1$ are absorbed into $2^{o(n)}$. The final result is obtained using \cref{eqn: mutdconv}.
\end{proof}

\subsection{A random coset Hessian}
\label{sec:linear-fibers}
Fix $0\le\chi\le1/2$.
We split each parity class $u \in \F_2^n$ into $h=\lfloor\chi n\rfloor$ and $\ell=n-h$ bits in random basis.
For $P\in\GL_n(\F_2)$, $u\in\F_2^n$, and $X\in\cL^*$, define
\[
  Pu=(\alpha,\theta)\in\F_2^h\times\F_2^\ell,
  \qquad
  P^{-T}B^TX\bmod2
  =
  (J_P(X),V_P(X))
  \in\F_2^h\times\F_2^\ell.
\]
Define $\Lambda_j:=\{X\in\cL^*:J_P(X)=j\}$. The set $\Lambda_0$ is a sublattice of $\cL^*$, and each $\Lambda_j$ for $j\in \F_2^h$ is one of its cosets. Note that every $\Lambda_j$ is nonempty because the map $X\mapsto P^{-T}(B^TX\bmod2)$ from $\cL^*$ to $\F_2^n$ is surjective.
Therefore, $\Lambda_j$ can be written as $\Lambda_0 + v_j$ for some $v_j \in \Lambda_j$ thus the discrete Gaussian $D_{\Lambda_j,s}$ is well-defined. 

For $\theta\in\F_2^\ell$, we define the (sublattice) coset Hessian along $\Lambda_j$ analogous to \cref{lem:hessian-formulas} by
\begin{align}
    \label{eqn: cGcosetexpectation}
  \cG_{t,d,j}(\theta)
  :=
  -4\pi^2
  \E_{X\sim D_{\Lambda_j,\xi_t}}
  \left[
    XX^T(-1)^{\theta^T V_P(X)}
  \right]
\end{align}
based on the identity for $X \in \Lambda_j$
\begin{align}
    \label{eqn: innersep}
  u^T(B^TX\bmod2)
  =
  (Pu)^T (P^{-T}B^TX\bmod2)
  =
  \alpha^T J_P(X)+\theta^T V_P(X)
  = \alpha^T j + \theta^T V_P(X)
  \pmod2
\end{align}
For $u=P^{-1}(\alpha,\theta)$, \cref{eqn: cGcosetexpectation}
omits the factor $(-1)^{\alpha^Tj}$. This omission allows $\cG_{t,d,j}$ is defined independent of $\alpha$, yet we should be careful about that factor.

Let
$\zeta_{t,j}
:=
2^h\Prob_{X\sim D_{\cL^*,\xi_t}}[J_P(X)=j]$
be the probability, normalized by the uniform probability $2^{-h}$,
that a discrete Gaussian sample of width $\xi_t$ belongs to the
affine lattice coset $\Lambda_j$.
The following lemma shows that $\zeta_{t,j}$ is near 1, i.e.,
the distribution of $J_P(X)$ for the Discrete Gaussian sample $X$ is
pointwise exponentially close to uniform. The parameter condition of this lemma is intentionally general for the later application.

\begin{lemma}\label{lem:target-coset-flatness}
If
$\lambda=\lambda_1(\cL)\le d\le(1+1/n)\lambda$, $t>t_0/2$, and
$0\le\chi\le1/2$ and $\chi<1/2+g(t)$, then, 
with probability
$1-2^{-\Omega(n)}$ over the uniform random $P\in \GL_n(\F_2)$,
$\zeta_{t,j}=1+O(2^{-\Omega(n)})$
for every $j\in\F_2^h$.
\end{lemma}
\begin{proof}
Recall $2^h\ind_{\{y=j\}}
=\sum_{\alpha\in\F_2^h}(-1)^{\alpha\cdot(y+j)}$ for $y,j\in\F_2^h$.
Applying the identity with $y=J_P(X)$ and taking the expectation over
$X\sim D_{\cL^*,\xi_t}$ gives
\[
\zeta_{t,j}
=\sum_{\alpha\in\F_2^h}(-1)^{\alpha\cdot j}
\E[(-1)^{\alpha\cdot J_P(X)}].
\]
Let $\cK_P:=P^{-1}(\F_2^h\times\{0\})$
and
$\kappa=P^{-1}(\alpha,0)\in\cK_P$ for $\alpha\in\F_2^h$.
The Poisson representation in \cref{lem:poisson} gives
\[
F_{1/\xi_t}(B\kappa/2) = \E_{X\sim D_{\cL^*,\xi_t}}[e^{2\pi i\ip{X}{B\kappa/2}}] = \E_{X\sim D_{\cL^*,\xi_t}}[(-1)^{\kappa^T(B^TX\bmod2)}]
=\E_{X\sim D_{\cL^*,\xi_t}}[(-1)^{\alpha\cdot J_P(X)}].
\]
Plugging this equation to the summand in $\zeta_{t,j}$ above, we have
\[
  \zeta_{t,j}
  =
  \sum_{\alpha\in\F_2^h}
  (-1)^{\alpha\cdot j}
  F_{1/\xi_t}
  \left(
    BP^{-1}(\alpha,0)/2
  \right) = 1 + \sum_{\alpha\neq 0}
  (-1)^{\alpha\cdot j}
  F_{1/\xi_t}
  \left(
    BP^{-1}(\alpha,0)/2
  \right).
\]
By substituting $\kappa = P^{-1}(\alpha,0)$, we have $|\zeta_{t,j}-1|
  \le
  \sum_{\substack{\kappa\in\cK_P\\\kappa\ne0}}
  F_{1/\xi_t}(B\kappa/2)$ for all $j$.
We bound this term using \cref{lem:spectral-identities}. For each nonzero $\kappa \in \F_2^n$, the space $\cK_P$ for random $P$ is a uniform random $h$-dimensional subspace so that 
$\Prob_P[\kappa\in\cK_P]
=(2^h-1)/(2^n-1)
=2^{(\chi-1+o(1))n}$. This gives
\[
  \E_P
  \sum_{{\kappa\in\cK_P,\kappa\ne0}}
  F_{1/\xi_t}(B\kappa/2)
  \le
  2^{(\chi-1/2-g(t)+o(1))n}.
\]
The exponent is a negative constant by the assumption
$\chi<1/2+g(t)$. Markov's inequality shows that, except with
probability $2^{-\Omega(n)}$ over $P$,
\[
  \sum_{{\kappa\in\cK_P,\kappa\ne0}}
  F_{1/\xi_t}(B\kappa/2)
  \le
  2^{-\Omega(n)}.
\]
This proves the lemma, as the upper bound of $|\zeta_{t,j}-1|$ holds simultaneously for all $j.$
\end{proof}

The following lemma shows the sampling from $D_{\Lambda_j,\xi_t}$ can be done by sampling from the entire $D_{\cL^*,\xi_t}$ and then collecting the samples with $J_P(X)=j$ thanks to the almost uniformity.

\begin{lemma}\label{lem:affine-coset-sampling}
If $\lambda\le d\le (1+1/n) \lambda$, $t_0 <t<1/4$, and $0\le \chi\le 1/2$, there is a classical algorithm that for independent uniform
$P\in\GL_n(\F_2)$ and $j\in\F_2^h$, outputs $N_t=\lceil n^5\log n\cdot 2^{2tn}\rceil$ lattice vectors or aborts in time $2^{(\max\{1/2,\chi+2t\}+o(1))n}$ and space
$2^{n/2+o(n)}$.
It does not abort with probability $1-2^{-\Omega(n)}$ over $P$
and the sampler randomness, and in that case, the joint distribution of its outputs has statistical
distance at most $\exp(-\Omega(n^2))$ from
$D_{\Lambda_j,\xi_t}^{N_t}$.
\end{lemma}
\begin{proof}
We use \cref{thm:dgs-sampling} $\lceil4N_t2^{h-n/2}\rceil$ times sequentially, collect the samples satisfying $J_P(X)=j$, and discard the others.
By \cref{lem:target-coset-flatness}, except with probability
$2^{-\Omega(n)}$ over $P$,
$  \Prob_{X\sim D_{\cL^*,\xi_t}}[J_P(X)=j]
  =
  2^{-h}\left(1+O(2^{-\Omega(n)})\right)$
simultaneously for every $j\in\F_2^h$.
For every such $P$ and every $j$, the Chernoff bound in
\cref{lem:Chernoff} shows that fewer than $N_t$ samples are
accepted among $4N_t2^h$ independent samples with probability
$2^{-\Omega(N_t)}$.
The statistical distance from independent
$D_{\cL^*,\xi_t}$ samples changes this probability by at most
$\exp(-\Omega(n^2))$.
Therefore, the procedure aborts with probability
$2^{-\Omega(n)}$.
The time and space complexity is clear.
\end{proof}

\subsection{Estimating the coset Hessian}
\label{sec:direct-coset-estimation}

Fix $t_0=0.23147\ldots<t<1/4$ and $N_t=\lceil n^5\log n\cdot 2^{2tn}\rceil$.
For independent samples from the coset discrete Gaussian distribution
$X_1,\ldots,X_{N_t}\sim D_{\Lambda_j,\xi_t}$, define an estimation of \cref{eqn: cGcosetexpectation} analogous to \cref{eqn: Gesti} by
\begin{align}\label{eqn:affine-coset-estimator}
  \widehat{\cG}_{t,d,j}(\theta)
  :=
  -\frac{4\pi^2}{N_t}
  \sum_{i=1}^{N_t}
  X_iX_i^T(-1)^{\theta\cdot V_P(X_i)}
  \ind_{\{\norm{X_i}\le C_1\xi_t\sqrt n\}}
\end{align}
which can be computed using 
\cref{lem:affine-coset-sampling}.
We show that $\widehat{\cG}_{t,d,j}$ approximates ${\cG}_{t,d,j}$ similar to \cref{lem:hessian-estimation}.

\begin{lemma}\label{lem:affine-coset-estimation}
Let $\lambda=\lambda_1(\cL)$ and
$\lambda\le d\le(1+1/n)\lambda$. Choose independent uniform
$P\in\GL_n(\F_2)$ and $j\in\F_2^h$. For independent
$X_1,\ldots,X_{N_t}\sim D_{\Lambda_j,\xi_t}$, with probability
$1-2^{-\Omega(n)}$ over $P,j,X_1,\ldots,X_{N_t}$,
\[
  \max_{\theta\in\F_2^\ell}
  \opnorm{
    \widehat{\cG}_{t,d,j}(\theta)
    -
    \cG_{t,d,j}(\theta)
  }
  = O\left(\frac{\mu_{t,d}}{n}\right).
\]
\end{lemma}

\begin{proof}
By \cref{lem:target-coset-flatness}, except with probability
$2^{-\Omega(n)}$ over $P$, we have $\zeta_{t,j}\ge1/2$ for every
$j\in\F_2^h$. Fix such a matrix $P$.
Let $Y\sim D_{\cL^*,\xi_t}$ and
$X\sim D_{\Lambda_j,\xi_t}$. 
By
$\Prob[J_P(Y)=j]=2^{-h}\zeta_{t,j}$ and
\cref{lem:dgs-tail} with $c=1$, we have
\[
  \E\left[
    \norm{X}^2
    \ind_{\{\norm{X}>C_1\xi_t\sqrt n\}}
  \right]
  \le
  \frac{2^h}{\zeta_{t,j}}
  \E\left[
    \norm{Y}^2
    \ind_{\{\norm{Y}>C_1\xi_t\sqrt n\}}
  \right]
  \le
  2^{h+1}\xi_t^2n2^{-n}.
\]
With the same argument for \cref{eqn: lem35_opbound},
we have
\[
  \opnorm{
    \E[\widehat{\cG}_{t,d,j}(\theta)]
    -
    \cG_{t,d,j}(\theta)
  }
  \le
  8\pi^2\xi_t^2n2^{-\ell}
  \le
  \frac{\xi_t^4d^2}{n}2^{-tn}
\]
for all sufficiently large $n$, where we used
$\ell\ge n/2$, $t<1/4$, and
$\xi_t^2d^2=4nt\ln2/\pi$.

Each entry of the random matrix in
\cref{eqn:affine-coset-estimator} has absolute value at most
$4\pi^2C_1^2\xi_t^2n$. Hence, for fixed
$\theta\in\F_2^\ell$ and $a,b\in\{1,\ldots,n\}$,
Hoeffding's inequality gives
\[
  \Prob\left[
    \left|
      \left(
        \widehat{\cG}_{t,d,j}(\theta)
        -
        \E[\widehat{\cG}_{t,d,j}(\theta)]
      \right)_{a,b}
    \right|
    >
    \frac{\xi_t^4d^2}{n^2}2^{-tn}
  \right]
  \le
  2\exp\left(
    -\Omega\left(
      \frac{N_t\xi_t^4d^4}{n^6}2^{-2tn}
    \right)
  \right)
  =
  2e^{-\Omega(n\log n)}.
\]
Here we used
$N_t=\lceil n^5\log n\cdot2^{2tn}\rceil$ and
$\xi_t^2d^2=4nt\ln2/\pi$.
A union bound over the $2^\ell n^2$ choices of
$\theta,a,b$ shows that, with probability
$1-2^{-\Omega(n\log n)}$ over the samples,
\[
  \max_{\theta\in\F_2^\ell}
  \opnorm{
    \widehat{\cG}_{t,d,j}(\theta)
    -
    \E[\widehat{\cG}_{t,d,j}(\theta)]
  }
  \le
  \frac{\xi_t^4d^2}{n}2^{-tn} 
\]
The triangle inequality applied to
$\widehat{\cG}_{t,d,j}(\theta)-
\E[\widehat{\cG}_{t,d,j}(\theta)]$ and
$\E[\widehat{\cG}_{t,d,j}(\theta)]-
\cG_{t,d,j}(\theta)$ proves the lemma, after using $\mu_{t,d} = \xi_t^4d^22^{-tn+O(1)} $.
\end{proof}

\subsection{Coset Hessians and shortest vectors}
It remains to show how to relate the coset Hessians and the (direction) of shortest vectors as in \cref{lem:shortest-class-hessian}.
The following lemma extends \cref{eqn:cAdef} to the coset Hessian.

\begin{lemma}\label{lem:affine-coset-hessian}
  $\cG_{t,d,j}(\theta)
  -
  \frac1{\zeta_{t,j}}
  \sum_{\alpha\in\F_2^h}
  (-1)^{\alpha\cdot j}
  \cA_{t,d}\left(P^{-1}(\alpha,\theta)\right)
  \in\operatorname{span}\{\Id\}.$
\end{lemma}

\begin{proof}
Since
$\Prob[J_P(X)=j]=2^{-h}\zeta_{t,j}$, the definition of
$\cG_{t,d,j}(\theta)$ gives
\[
  \cG_{t,d,j}(\theta)
  =
  -\frac{4\pi^2 2^h}{\zeta_{t,j}}
  \E_{X\sim D_{\cL^*,\xi_t}}
  \left[
    XX^T(-1)^{\theta\cdot V_P(X)}
    \ind_{\{J_P(X)=j\}}
  \right].
\]
Observe that $2^h\ind_{\{y=j\}}
  =
  \sum_{\alpha\in\F_2^h}
  (-1)^{\alpha\cdot(y+j)}$ for every $y,j\in\F_2^h$.
Indeed, if $y=j$, every summand is one; otherwise, the summands can
be paired with opposite signs.
Applying this identity with $y=J_P(X)$ and interchanging the finite
sum and expectation gives
\begin{align}
    \label{eqn: cGdual}
  \cG_{t,d,j}(\theta)
  =
  \frac1{\zeta_{t,j}}
  \sum_{\alpha\in\F_2^h}
  (-1)^{\alpha\cdot j}
  \cG_{t,d}\left(P^{-1}(\alpha,\theta)\right)
\end{align}
where we put $u=P^{-1}(\alpha,\theta)$ and use
\cref{eqn: innersep} to derive
\[
(-1)^{\theta\cdot V_P(X)} \cdot (-1)^{\alpha \cdot J_P(X) + \alpha \cdot j}  = (-1)^{\alpha\cdot j} \cdot (-1)^{\alpha\cdot J_P(X) + \theta\cdot V_P(X)} = (-1)^{\alpha \cdot j} \cdot (-1)^{u^T (B^T X \bmod 2)}.
\]
The term $(-1)^{u^T (B^T X \bmod 2)}$ is exactly the term in
\cref{lem:hessian-formulas}, proving \cref{eqn: cGdual}.
The final inclusion is proven by using \cref{lem:hessian-formulas} on 
$\cG_{t,d}(P^{-1}(\alpha,\theta))-\cA_{t,d}(P^{-1}(\alpha,\theta))\in\operatorname{span}\{\Id\}$ for each $\alpha.$
\end{proof}

This lemma gives some intuition why the approach in this section works. 
If one of $P^{-1}(\alpha,\theta)=u$ is such that $Bu+2\cL$ contains a shortest vector $v$, the corresponding $\cA_{t,d}(P^{-1}(\alpha,\theta))$ must be approximately proportional to $vv^T$. 
We will prove that the sum of the remaining terms is
exponentially smaller with high probability over $P$ and $j$.
Therefore, some eigenvector of $\cG_{t,d,j}(\theta)$ must be approximately proportional to $v$, as
the multiple of $I_n$ does not affect the eigen vectors and the order of eigenvalues.

The coefficient of the first term is
$(-1)^{\alpha_*\cdot j}$, so it may determine either the largest or
the smallest eigenvalue. We therefore consider both extreme
eigenvalues. Similarly to \cref{lem:hessian-recovery}, we obtain the
following lemma.

\begin{lemma}\label{lem:affine-coset-recovery}
Let $t_0<t<1/4$ be fixed, let $0\le\chi\le1/2$, and let
$\lambda=\lambda_1(\cL)\le d\le(1+1/n)\lambda$.
Let $v\in\cL$ satisfy $\norm{v}=\lambda$, and let
$u_*\in\F_2^n$ satisfy $v\in Bu_*+2\cL$.
Choose independent uniform $P\in\GL_n(\F_2)$ and
$j\in\F_2^h$, and write $Pu_*=(\alpha_*,\theta_*)$.
With probability $1-2^{-\Omega(n)}$ over $P$ and $j$, the following
holds. Suppose that the conclusion of
\cref{lem:affine-coset-estimation} holds. Let $q_+$ and $q_-$ be
unit eigenvectors corresponding to the largest and smallest
eigenvalues of $\widehat{\cG}_{t,d,j}(\theta_*)$, respectively.
Then there are $s\in\{+,-\}$ and $\tau\in\{-1,1\}$ such that
\[
  \norm{q_s-\tau v/\lambda}
  =
  O(n^{-1/2}).
\]
Consequently, the $n^{-1/3}$-$\BDD$ query at $dq_s$ returns
$\tau v$ for all sufficiently large $n$.
\end{lemma}

\begin{proof}
Let $\widehat v=v/\lambda$ and recall
\cref{lem:shortest-class-hessian} stating that
$\cA_{t,d}(u_*)
  =
  \mu_{t,d}\widehat v\widehat v^T
  +
  \cR_{t,d}(u_*)$
for ${\opnorm{\cR_{t,d}(u_*)}}/{\mu_{t,d}}
  =
  2^{-\Omega(n)}$, where $\mu_{t,d}=\xi_t^4 \lambda^2 2^{-tn+O(1)}$.

We first bound the contribution from $\alpha\ne\alpha_*$. For fixed
$P$, expanding the squared Frobenius norm using $\Frob{X}^2=\operatorname{Tr}(X^TX)$ for the real matrix $X$, 
\begin{align*}
    \Frob{
    \sum_{\alpha\ne\alpha_*}
    (-1)^{j\cdot(\alpha-\alpha_*)}
    \cA_{t,d}\left(P^{-1}(\alpha,\theta_*)\right)
  }^2
  =\sum_{\alpha,\alpha' \neq \alpha_*} (-1)^{j\cdot(\alpha-\alpha')} \operatorname{Tr}\left(
  \cA_{t,d}\left(P^{-1}(\alpha',\theta_*)\right)^T \cA_{t,d}\left(P^{-1}(\alpha,\theta_*)\right)
  \right)
\end{align*}
and taking the expectation
over uniform $j$ using $\E_j[(-1)^{j\cdot(\alpha-\alpha')}]=0$ for $\alpha\neq \alpha'$
gives
\[
  \E_j\Frob{
    \sum_{\alpha\ne\alpha_*}
    (-1)^{j\cdot(\alpha-\alpha_*)}
    \cA_{t,d}\left(P^{-1}(\alpha,\theta_*)\right)
  }^2
  =
  \sum_{\alpha\ne\alpha_*}
  \Frob{
    \cA_{t,d}\left(P^{-1}(\alpha,\theta_*)\right)
  }^2.
\]

Now we are taking the expectation over $P$. 
For each fixed $u\ne u_*$, the matrix $\cA_{t,d}(u)$ occurs in the
sum on exactly when
$P(u-u_*)\in\F_2^h\times\{0\}$. Since $P(u-u_*)$ is uniform over
$\F_2^n\setminus\{0\}$, the probability that $P(u-u_*)$ is included in $\F_2^h\times\{0\}$ is exactly $(2^h-1)/(2^n-1)$.
Consequently, \cref{lem:hessian-square-mass} gives
\[
  \E_{P,j}\Frob{
    \sum_{\alpha\ne\alpha_*}
    (-1)^{j\cdot(\alpha-\alpha_*)}
    \cA_{t,d}\left(P^{-1}(\alpha,\theta_*)\right)
  }^2
  \le
    2^{h-n}  \cdot 
    (\mu_{t,d}^2 2^{2tn-g(t)n+o(n)})
    =\mu_{t,d}^2 2^{(\chi -1 -g(t) +2t+o(1))n}
\]
where we use $h=\lfloor \chi n \rfloor$.
Since $\chi\le1/2$, $t<1/4$, and $g(t)>0$, the exponent satisfies
$\chi-1-g(t)+2t<0$. Markov's inequality shows that
\begin{align}\label{eqn: Frobremainderbound}
  \Frob{
    \sum_{\alpha\ne\alpha_*}
    (-1)^{j\cdot(\alpha-\alpha_*)}
    \cA_{t,d}\left(P^{-1}(\alpha,\theta_*)\right)
  }
  \le
  \mu_{t,d}2^{-\Omega(n)}
\end{align}
holds except
with probability $2^{-\Omega(n)}$ over $P$ and $j$, and the same bound holds for the operator norm as well.

For the hidden sign $\sigma=(-1)^{\alpha_*\cdot j}$, we can decompose $E:=\sigma\widehat{\cG}_{t,d,j}(\theta_*)
  - \frac{\mu_{t,d}}{\zeta_{t,j}}
  \widehat v\widehat v^T - a\Id$ using \cref{lem:affine-coset-hessian} 
for some $a$
and bound its operator norm by
\[
  \opnorm{\widehat{\cG}_{t,d,j}(\theta_*) - {\cG}_{t,d,j}(\theta_*)}
  +
  \frac{1}{\zeta_{t,j}}
  \left(\opnorm{\sum_{\alpha\ne\alpha_*}
    (-1)^{j\cdot(\alpha-\alpha_*)}
    \cA_{t,d}\left(P^{-1}(\alpha,\theta_*)\right)}
  +\opnorm{\cR_{t,d}(u_*)}\right) = O\left(\mu_{t,d}n^{-1}\right)
\]
where we use \cref{lem:affine-coset-estimation} and \cref{eqn: Frobremainderbound} together with
$\zeta_{t,j}=1+O(2^{-\Omega(n)})$ from \cref{lem:target-coset-flatness}.

Let $q$ be the unit eigenvector corresponding to the largest eigenvalue of $\sigma\widehat{\cG}_{t,d,j}(\theta_*)$, which must be one of $q_+$ (if $\sigma=1$) or $q_-$ (if $\sigma=-1$).
Hence
$q^T\sigma\widehat{\cG}_{t,d,j}(\theta_*)q
\ge
\widehat v^T\sigma\widehat{\cG}_{t,d,j}(\theta_*)\widehat v$.
Plugging $\sigma\widehat{\cG}_{t,d,j}(\theta_*) =   \frac{\mu_{t,d}}{\zeta_{t,j}}\widehat v\widehat v^T + a\Id + E$, we obtain the following inequality
\[
  \frac{\mu_{t,d}}{\zeta_{t,j}}
  \left(1-\ip{q}{\widehat v}^2\right)
  \le
  2\opnorm{E} = O\left(\mu_{t,d}n^{-1}\right).
\]
Thus $1-\ip{q}{\widehat v}^2=O(n^{-1})$. Choosing
$\tau\in\{-1,1\}$ such that
$\ip{q}{\tau\widehat v}\ge0$ gives
$\norm{q-\tau\widehat v}=O(n^{-1/2})$.

Finally,
$\norm{dq-\tau v}
\le d\norm{q-\tau\widehat v}+d-\lambda
=O(\lambda n^{-1/2})
<\lambda n^{-1/3}$
for all sufficiently large $n$. The final claim follows from
\cref{thm:bdd-preprocessing}.
\end{proof}

\subsection{The affine-coset algorithm}
\label{sec:affine-coset-algorithm}
This section presents the algorithm using the affine cosets. 
The overall time complexity is $
2^{0.7314 n}\approx 2^{(0.5+t)n+o(n)} $ by balancing $2^{\max(0.5,\chi+2t) n +o(n)}$ of the DGS sampling from \cref{lem:affine-coset-sampling} and $2^{\ell+o(n)} = 2^{(1-\chi)n+o(n)}$ of the number of $\theta \in \F_2^\ell$, where $h=\lfloor \chi n \rfloor$, $t_0 =0.2314\ldots < t < 1/4$ and $N_t=2^{(2t+o(1))n}$. This can be obtained using the Walsh-Hadamard transform as before, and the other steps take much smaller time.

\begin{algorithmblock}{Affine-coset Hessian $\SVP$}
\label{alg:direct-row}
  \item Construct the preprocessing $n^{-1/3}$-$\BDD$ data and the
        scale grid from \cref{alg:direct-exhaustive}.
        Choose $t_0<t<1/4$, and put $\chi=1/2-t$.

  \item At every scale $d$, choose independent uniform
        $P\in\GL_n(\F_2)$ and $j\in\F_2^h$.
        Apply the sampling algorithm from
        \cref{lem:affine-coset-sampling}.
        Abort the scale $d$ unless $N_t$ samples from
        $D_{\Lambda_j,\xi_t}$ are obtained.

  \item Let $b=\lfloor n/2\rfloor$. For every accepted sample $X_i$,
        define
        \[
          V_P(X_i)=(k_i',k_i'')
          \in\F_2^b\times\F_2^{\ell-b}, \qquad
          W_i:=
          -\frac{4\pi^2}{N_t}X_iX_i^T
          \ind_{\{\norm{X_i}\le C_1\xi_t\sqrt n\}}.
        \]
        For each $\theta''\in\F_2^{\ell-b}$, do the following:
        \begin{enumerate}
            \item Construct
        the array $
          A_{\theta''}(x)
          :=
          \sum_{\substack{1\le i\le N_t\\k_i'=x}}
          (-1)^{\theta''\cdot k_i''}W_i$ for all $x\in\F_2^b$.
          \item Apply the matrix-valued Walsh-Hadamard transform to obtain for every $\theta'\in\F_2^b$:
        \begin{align}\label{eqn:sequential-coset-wht}
          \sum_{x\in\F_2^b}
          (-1)^{\theta'\cdot x}A_{\theta''}(x)
          =
          \sum_{i=1}^{N_t}
          (-1)^{\theta'\cdot k_i'
          +\theta''\cdot k_i''}W_i
          =
          \widehat{\cG}_{t,d,j}(\theta',\theta'').
        \end{align}
        \item For each output in \cref{eqn:sequential-coset-wht},
        compute unit eigenvectors $q_+$ and $q_-$ corresponding to
        its largest and smallest eigenvalues and query $\BDD$ at $dq_\pm$. Store a nonzero output in $\cL$.
        Keep only the shortest vector found, and discard the other vectors and temporal data.
        \end{enumerate}


  \item Return the shortest vector over all scales.
\end{algorithmblock}


\begin{theorem}\label{thm:direct-row}
\cref{alg:direct-row} solves
Search-$\SVP$ with constant success probability in time
$  2^{0.7314n+o(n)}$
and space
$2^{n/2+o(n)}$.
\end{theorem}

\begin{proof}
Suppose
$\lambda_1(\cL)\le d\le(1+1/n)\lambda_1(\cL)$ and
let $v$ a shortest vector. 
Choose $u_*\in\F_2^n$ such that
$v\in Bu_*+2\cL$, write $Pu_*=(\alpha_*,\theta_*)$, and write
$\theta_*=(\theta_*',\theta_*'')\in
\F_2^b\times\F_2^{\ell-b}$.

By \cref{lem:affine-coset-sampling}, the algorithm obtains the
specified samples from $D_{\Lambda_j,\xi_t}$. By
\cref{eqn:sequential-coset-wht}, the iteration indexed by
$\theta_*''$ computes
$\widehat{\cG}_{t,d,j}(\theta_*',\theta_*'')
=\widehat{\cG}_{t,d,j}(\theta_*)$.
The conclusion of \cref{lem:affine-coset-estimation} holds
simultaneously for every $\theta\in\F_2^\ell$, and
\cref{lem:affine-coset-recovery} shows the correctness.
The statistical distance in \cref{lem:affine-coset-sampling} only changes
the success probability by $\exp(-\Omega(n^2))$.

We analyze the time and space complexity. 
Sampling takes
$2^{(\max\{1/2,\chi+2t\}+o(1))n}$ time by \cref{lem:affine-coset-sampling}. There are
$2^{\ell-b}$ sequential iterations. In each iteration, constructing
the array takes $N_t\poly(n)$ time, and its Walsh-Hadamard
transform, eigenvector computations, and $2^{b+1}$ $\BDD$ queries
take $2^{b+o(n)}$ time. Since
$N_t=2^{(2t+o(1))n}<2^{b+o(n)}=2^{0.5n+o(n)}$, their total time is $2^{\ell +o(n)}$.
Since $\ell=(1-\chi)n+O(1)$ and $\chi=1/2-t$, the time exponent is $\max\{\chi+2t,1-\chi\}=1/2+t\to0.7314$ for $t\to t_0$. The space complexity is dominated by the number of accepted samples and the preprocessing $\BDD$ data, which is $2^{n/2+o(n)}$.
\end{proof}

\section{Importance Sampling on an Affine Lattice Coset}
\label{sec:best}
This section improves the algorithm in \cref{sec:random-row} by
sampling from $D_{\Lambda_j,\xi_R}$ and estimating the Hessian
defined using the smaller width $\xi_r$.

Fix constants $0<r<R$ with $r<1/4$ and $0<\chi<1/2$.
Let $h=\lfloor\chi n\rfloor$ and $\ell=n-h$ as in
\cref{sec:linear-fibers}.
Recall $g(t)=\frac12\log_2(t/t_0)$.
We use the parameters $R$ and $r$ for the source and target
widths, respectively.

\subsection{Lattice points in a shortest parity class}
\label{sec:parity-shell}
We first study the number of lattice points in the parity classes to improve some inequalities and conditions in some lemmas.
For $0<\varphi<\pi/2$, define the Kabatiansky-Levenshtein constant
\[
  B_{\rm KL}(\varphi)
  :=
  \frac{1+\sin\varphi}{2\sin\varphi}
  \log_2\frac{1+\sin\varphi}{2\sin\varphi}
  -
  \frac{1-\sin\varphi}{2\sin\varphi}
  \log_2\frac{1-\sin\varphi}{2\sin\varphi}.
\]
Note that $\beta$ in \cref{thm:lattice-points} is such that $\log_2 \beta = B_{\rm KL}(\pi/3)$ as shown in \cite{EPRINT:PujSte09}.
We use the following bound \cite{KL78}.

\begin{theorem}\label{thm:functional-kl}
Let $0<\varphi<\pi/2$ be fixed. If
$S\subseteq\{x\in\R^n:\norm{x}=1\}$ and the angle between every
two distinct vectors in $S$ is at least $\varphi$, then
$|S|\le2^{(B_{\rm KL}(\varphi)+o(1))n}$. The $o(1)$ term is uniform
when $\varphi$ ranges over a compact subinterval of $(0,\pi/2)$.
\end{theorem}

The counting argument in the following lemma is a 
variant of \cite[Lemma 3]{EPRINT:PujSte09}. The restriction to
$v+2\cL$ gives the stronger separation used below.
\begin{lemma}\label{lem:parity-separation}
Let $v$ be a shortest vector and put $\lambda=\norm{v}$.
Every $w\in v+2\cL\setminus\{\pm v\}$ satisfies
$\norm{w}\ge\sqrt3\lambda$. Moreover, for every fixed
$C\ge\sqrt3$, uniformly for $\sqrt3\le x\le C$,
\[
  \left|
    \left\{
      w\in v+2\cL\setminus\{\pm v\}:
      \norm{w}\le x\lambda
    \right\}
  \right|
  \le
  2^{(B_{\rm KL}(\arccos(1-2/x^2))+o(1))n}.
\]
\end{lemma}

\begin{proof}
Write $w=2y-v$. If $w\ne\pm v$, then $y$ and $y-v$ are nonzero
lattice vectors. The parallelogram identity gives
$\norm{w}^2+\lambda^2=2\norm{y}^2+2\norm{y-v}^2\ge4\lambda^2$.

For distinct $w,w'\in v+2\cL$, the vector $(w-w')/2$ is a nonzero
lattice vector, and hence $\norm{w-w'}\ge2\lambda$. For each vector
with $\norm{w}\le x\lambda$, append the coordinate
$\sqrt{x^2\lambda^2-\norm{w}^2}$. The resulting vectors in
$\R^{n+1}$ have norm $x\lambda$ and pairwise distance at least
$2\lambda$. After normalization, their pairwise angles are at least
$\arccos(1-2/x^2)$. The result follows from
\cref{thm:functional-kl}.
\end{proof}

We improve the bound on $\rho_{2\sqrt2 /\xi_t}(\cL)\le 1+2^{(1-g(t)+o(1))n}$ in \cref{cor:working-scale} (i.e. \cref{eqn:simplerho_g}) using a more fine-grained counting argument with the vectors in each coset modulo $2\cL$.
For $x\ge1$, define
\[
  K_2(x)
  :=
  \begin{cases}
    0,&1\le x\le\sqrt2,\\
    B_{\rm KL}(\arccos(1-2/x^2)),&x>\sqrt2,
  \end{cases}
\]
Combining \cref{thm:functional-kl} with \cref{thm:lattice-points}, the number of nonzero vectors
in $\cL$ of norm at most $x\lambda$ is
\[
  N_{\cL}(x\lambda_1(\cL))
  \le 2^{\left(
    \min\{\log_2\beta+\log_2x,1+K_2(x)\}+o(1)
  \right)n}.
\]

Recall $g(t)=\frac{1}{2}\log_2(t/t_0)$ for $t_0=0.23147\ldots.$
Define\footnote{The inequality is not obvious. We sketch the proof for $R < 1/\ln 2$ which is only relevant to our analysis. In this case, dropping the second term from the minimum gives
$g_2(R)\ge1-\sup_{x\ge1}
(\log_2\beta+\log_2x-Rx^2/2)$. The supremum is attained at
$x=1/\sqrt{R\ln2}$, and the resulting lower bound equals $g(R)$.}
\begin{align}\label{eqn:refined-source-exponent}
  g_2(R)
  :=
  1-
  \sup_{x\ge1}
  \left(
    \min\left\{
      \log_2\beta+\log_2x,
      1+K_2(x)
    \right\}
    -\frac R2x^2
  \right) \ge g(R).
\end{align}

\begin{lemma}\label{lem:parity-source-mass}
Let 
$\lambda=\lambda_1(\cL)\le d\le(1+1/n)\lambda$. For every fixed $R>0$,
$  \rho_{2\sqrt2/\xi_R}(\cL\setminus\{0\})
  \le
  2^{(1-g_2(R)+o(1))n}.$
\end{lemma}
Note that the proof of this lemma is rather complicated and somewhat independent to the other arguments. The readers may use a simpler bound with $g$. All arguments in this section works well by using $g$ in place of $g_2$, which gives a only slightly bad complexity of $2^{0.6040n+o(n)}$ time and space. See \cref{rem: nonopt}.

\begin{proof}
We first bound the number of short vectors in one coset modulo
$2\cL$. Fix a coset $C\in\cL/2\cL$ and put
\[
  S_C(x)
  :=
  \left\{
    w\in C\setminus\{0\}:
    \norm{w}\le x\lambda
  \right\},
  \qquad
  m:=|S_C(x)|.
\]
For every $w\in S_C(x)$, define
  $u_w
  :=
  \frac1{x\lambda}
  \left(
    w,
    \sqrt{x^2\lambda^2-\norm{w}^2}
  \right)
  \in\R^{n+1}.$
Then $\norm{u_w}=1$. If $w,w'\in S_C(x)$ are distinct, then
$(w-w')/2$ is a nonzero vector in $\cL$, and hence
$\norm{w-w'}\ge2\lambda$. It follows that
\[
  \norm{u_w-u_{w'}}
  \ge
  \frac{\norm{w-w'}}{x\lambda}
  \ge
  \frac2x \quad \Longrightarrow \quad
  \ip{u_w}{u_{w'}}
  =
  1-\frac12\norm{u_w-u_{w'}}^2
  \le
  1-\frac2{x^2}.
\]

Suppose first that $1\le x\le\sqrt2$. Then
$\ip{u_w}{u_{w'}}\le0$ for distinct $w,w'$. Let
$G=(\ip{u_w}{u_{w'}})_{w,w'\in S_C(x)}$ be the Gram matrix of
these unit vectors. Since $\operatorname{rank}(G)\le n+1$ and
$\operatorname{tr}(G)=m$, the Cauchy-Schwarz inequality for the
eigenvalues of $G$ gives
\[
  \Frob{G}^2
  \ge
  \frac{\operatorname{tr}(G)^2}
       {\operatorname{rank}(G)}
  \ge
  \frac{m^2}{n+1}.
\]
On the other hand,
$-1\le\ip{u_w}{u_{w'}}\le0$ for $w\ne w'$, and hence
$\ip{u_w}{u_{w'}}^2\le-\ip{u_w}{u_{w'}}$.
Moreover,
\[
  0
  \le
  \norm{\sum_{w\in S_C(x)}u_w}^2
  =
  m+
  \sum_{\substack{w,w'\in S_C(x)\\w\ne w'}}
  \ip{u_w}{u_{w'}}.
\]
Consequently,
\[
  \Frob{G}^2
  =
  m+
  \sum_{\substack{w,w'\in S_C(x)\\w\ne w'}}
  \ip{u_w}{u_{w'}}^2
  \le
  m-
  \sum_{\substack{w,w'\in S_C(x)\\w\ne w'}}
  \ip{u_w}{u_{w'}}
  \le
  2m.
\]
Combining the two bounds on $\Frob{G}^2$ gives
$m\le2(n+1) = 2^{o(n)}$.

Now consider $x>\sqrt2$. The preceding inner-product
bound shows that the angle between every two distinct vectors
$u_w,u_{w'}$ is at least
$\arccos(1-2/x^2)$. Applying
\cref{thm:functional-kl} in dimension $n+1$ gives
\[
  |S_C(x)|
  \le
  2^{(K_2(x)+o(1))n}.
\]
Combining two cases, each coset modulo $2\cL$ contains
$2^{(K_2(x)+o(1))n}$ vectors of norm
$\le x\lambda$ for every fixed $x\ge1$.
Summing over $2^n$ cosets modulo $2\cL$, giving $2^{(1+K_2(x)+o(1))n}$. Together with \cref{thm:lattice-points}, we have
\[
  N_{\cL}(x\lambda)
  \le
  2^{\left(
    \min\{
      \log_2\beta+\log_2x,
      1+K_2(x)
    \}
    +o(1)
  \right)n} = 2^{A(x)n +o(n)}
\]
for $A(x)=\min\{
    \log_2\beta+\log_2x,
    1+K_2(x)
  \}$.
We now estimate the Gaussian mass.
Put $x_k=(1+1/n)^k$ and
$  S_k
  :=
  \left\{
    w\in\cL:
    x_k\lambda\le\norm{w}<x_{k+1}\lambda
  \right\}.$
Using
$  \rho_{2\sqrt2/\xi_R}(w)
  =
  2^{-(R/2)n\norm{w}^2/d^2},$
we have
\[
  \rho_{2\sqrt2/\xi_R}(\cL\setminus\{0\})
  \le
  \sum_{k\ge0}
  N_{\cL}(x_{k+1}\lambda)
  2^{-(R/2)n(\lambda^2/d^2)x_k^2}.
\]

Fix $\delta>0$ and a constant $X_0>1$. Partition
$[1,X_0]$ into finitely many intervals $[a,b]$ such that
$R(b^2-a^2)/2\le\delta$. For each interval, consider 
$w\in\cL$ satisfying
$a\lambda\le\norm{w}<b\lambda$. Their number is at most
$2^{(A(b)+o(1))n}$, while each of their Gaussian weights is at most
$2^{-(R/2)n(\lambda^2/d^2)a^2}$. Hence their total contribution is
at most
\[
  2^{\left(
    A(b)-Ra^2/2+o(1)
  \right)n}
  \le
  2^{\left(
    \sup_{x\ge1}F_R(x)+\delta+o(1)
  \right)n}
\]
where $F_R(x):=A(x)-\frac R2x^2$. The last upper bound is independent of the interval, and the number of intervals in the decomposition of $[1,X_0]$ is independent of $n$, so their total
contribution satisfies the same bound.

It remains to consider $x_k>X_0$. Here we use only
\cref{thm:lattice-points}. Apart from its uniform
$2^{o(n)}$ factor, the upper bound for the $k$-th shell is
\[
  U_k
  :=
  \beta^n x_{k+1}^n
  2^{-(R/2)n(\lambda^2/d^2)x_k^2}\quad\Longrightarrow \quad
  \frac{U_{k+1}}{U_k}
  \le
  e\,
  2^{-R(\lambda^2/d^2)x_k^2(1+1/(2n))}
\]
since $x_{k+1}=(1+1/n)x_k$.
Choose $X_0$ sufficiently so large that this ratio is at most $1/2$
for all sufficiently large $n$ and simultaneously
the first such shell is bounded by
$2^{(\sup_{x\ge1}F_R(x)+\delta+o(1))n}$, which is possible because
$\log_2\beta+\log_2x-Rx^2/2$ tends to $-\infty$ as
$x\to\infty$. The bounds for the remaining shells form a geometric
series and hence have the same total upper bound.

Letting $\delta$ tend to zero gives
\[
  \rho_{2\sqrt2/\xi_R}(\cL\setminus\{0\})
  \le
  2^{(\sup_{x\ge1}F_R(x)+o(1))n}
  =
  2^{(1-g_2(R)+o(1))n},
\]
where the last equality follows from
\cref{eqn:refined-source-exponent}.
\end{proof}

The following lemma shows a spectral phenomena for the Hessian, or more precisely its translated version $\cA$ in 
\cref{eqn: cApositive}, similar to \cref{lem:shortest-class-hessian} even if $r<t_0=0.23147\ldots$.
We choose the lower bound $9/50=0.18$ only to ensure
\begin{align}\label{eqn:bKLineq}
  \sup_{x\ge\sqrt3}
  \left\{
    B_{\rm KL}(\arccos(1-2/x^2))
    -r(x^2-1)
  \right\}
  <-\frac1{10},
\end{align}
which follows by standard calculus.
Also note that $B_{\rm KL}(\arccos(1-2/x^2))$ is Lipschitz in the interval $[\sqrt 3,7]$. We define $\mu_{r,d} = \frac{
    2\pi^2\xi_r^4\lambda^2
  }{
    \rho_{1/\xi_r}(\cL)
  }
  e^{-\pi\xi_r^2\lambda^2/4}$ for any $r$.


\begin{lemma}\label{lem:parity-rank-one}
Let $\lambda=\lambda_1(\cL)\le d\le(1+1/n)\lambda$.
Let $v\in\cL$ satisfy $\norm{v}=\lambda$, let
$u\in\F_2^n$ satisfy $v\in Bu+2\cL$, and put
$\widehat v=v/\lambda$. If $r\ge 9/50$, then
\[
  \Frob{
    \cA_{r,d}(u)
    -
    \mu_{r,d}\widehat v\widehat v^T
  }
  \le
  \mu_{r,d}2^{-\Omega(n)}.
\]
\end{lemma}

\begin{proof}
The terms $w=\pm v$ in \cref{eqn: cApositive} give
$\mu_{r,d}\widehat v\widehat v^T$. The sum of the other terms is
positive semidefinite, so its Frobenius norm is at most its trace.
Consequently,
\[
  \frac{
    \Frob{
      \cA_{r,d}(u)-\mu_{r,d}\widehat v\widehat v^T
    }
  }{\mu_{r,d}}
  \le
  \frac12
  \sum_{\substack{w\in v+2\cL\\w\ne\pm v}}
  \left(\frac{\norm{w}}\lambda\right)^2
  2^{-rn(\lambda^2/d^2)(({\norm{w}}/\lambda)^2-1)}.
\]

Put $x_k:=\sqrt3(1+1/n)^k$ and let
$  S_k
  :=
  \left\{
    w\in v+2\cL\setminus\{\pm v\}:
    x_k\le \frac{\norm{w}}\lambda<x_{k+1}
  \right\}.$
  Then
\[
  \frac12
  \sum_{\substack{w\in v+2\cL\\w\ne\pm v}}
  \left(\frac{\norm{w}}\lambda\right)^2
  2^{-rn(\lambda^2/d^2)(({\norm{w}}/\lambda)^2-1)}
  \le
  \frac12
  \sum_{k\ge0}
  |S_k|x_{k+1}^2
  2^{-rn(\lambda^2/d^2)(x_k^2-1)}.
\]
We will show that the summations for $x_k<7$ and $x_k\ge7$ both are bounded above by $2^{-\Omega(n)}.$

First, if $x_k<7$, \cref{lem:parity-separation} gives the bound 
\[
\frac 1n\log_2 \left(|S_k|x_{k+1}^2
  2^{-rn(\lambda^2/d^2)(x_k^2-1)}\right)
  \le 
      (B_{\rm KL}(\arccos(1-2/x_{k+1}^2))
      +o(1)) -
      \frac{r\lambda^2 (x_k^2-1)}{d^2}
\]
and using the fact that $B_{\rm KL}(\arccos(1-2/x^2))$ is Lipschitz, $x_{k+1}=x_k + x_k/n = x_k+O(1/n)$ for $x_k<7$, and $d/\lambda=1+O(1/n)$ gives an upper bound
\[
  B_{\rm KL}(\arccos(1-2/x_k^2))
  -
  r(x_k^2-1)
  +
  o(1)
  \le
  -\frac17+o(1)
\]
where we use 
\cref{eqn:bKLineq},
proving $\frac12
  |S_k|x_{k+1}^2
  2^{-rn(\lambda^2/d^2)(x_k^2-1)}
  \le
  2^{(-1/7+o(1))n}.$
  Since there are at most $O(n)$ $k$'s such that $x_k<7$, it gives the desire bound for the first summand over $x_k<7$.

Now consider $x_k\ge7$. Since
$d\le(1+1/n)\lambda$, $r\lambda^2/d^2 \ge 1/6$ for sufficiently large $n.$
\cref{thm:lattice-points} on $|S_k|$ gives
\[
\frac 1n\log_2 \left(|S_k|
  2^{-rn(\lambda^2/d^2)(x_k^2-1)}\right)
  \le
  \log_2\beta+\log_2x_{k+1}-\frac{ x_k^2-1}{6}+o(1)
\]
Since $x_{k+1}=(1+1/n)x_k\le2x_k$, the exponent is at most
\[
  \log_2\beta+\log_2(2x_k)
  -
  \frac16(x_k^2-1)
  +
  o(1) \le -\frac{x_k^2}{14} +o(1)
\]
which holds for $x_k\ge 7.$
Applying this inequality with $x=x_k$ proves
\[
  \frac12
  |S_k|x_{k+1}^2
  2^{-rn(\lambda^2/d^2)(x_k^2-1)}
  \le
  \frac{x_{k+1}^2}{2}
  2^{-nx_k^2/14+o(n)}.
\]
This term (ignoring the uniform $o(1)$ factor) decays (super-)exponentially fast,\footnote{In particular, for $x_k\ge 7$, $(x_{k+2}^2 2^{-nx_{k+1}^2/14})/(x_{k+1}^2 2^{-nx_k^2/14})=(1+1/n)^2 2^{-nx_{k+1}^2 (2/n-1/n^2) /14} \le 1/2$ holds.} thus the asymptotic bound of the summation of the above term for $k$ such that $x_k \ge 7$ becomes $2^{-\Omega(n)}$.
\end{proof}

\subsection{Sampling (wider) discrete Gaussian from an affine lattice coset}
\label{sec:source-coset}
This section shows that sampling from the discrete Gaussian $D_{\Lambda_j , \xi_R}$ for a large $\xi_R$ is more efficient than the one in \cref{lem:affine-coset-sampling} by removing the multiplicative factor $2^\chi$ in the time complexity.

For the fixed $d$ and $R$, recall that
$\xi_R=\sqrt{4nR\ln2/(\pi d^2)}$.
For $P\in\GL_n(\F_2)$ and $j\in\F_2^h$, recall from
\cref{sec:linear-fibers} that
$  \Lambda_j
  :=
  \{X\in\cL^*:J_P(X)=j\}.$
In particular, $\Lambda_0$ is a sublattice of $\cL^*$ and
$\Lambda_j$ is a coset of $\Lambda_0$.
The following lemma shows some Gaussian mass bound from \cref{cor:working-scale} works for $\Lambda_0$ in expectation.


\begin{lemma}\label{lem:refined-source-coset-smoothing}
Let $\lambda=\lambda_1(\cL)$ and
$\lambda\le d\le(1+1/n)\lambda$. If $R>t_0$ and
$\chi<g_2(R)$, then
$  \rho_{\sqrt2/\xi_R}(\Lambda_0^*\setminus\{0\})
  \le
  2^{-\Omega(n)}$
with probability $1-2^{-\Omega(n)}$ over uniform
$P\in\GL_n(\F_2)$. Consequently,
$\xi_R>\sqrt2\,\eta_{1/2}(\Lambda_0)$.
\end{lemma}
\begin{proof}
Let
$\cK_P=P^{-1}(\F_2^h\times\{0\})$.
We first observe that $\Lambda_0^*=\cL+B\cK_P/2$. 
For every nonzero $\kappa\in\F_2^n$,
it is easy to see that
$\Prob_P[\kappa\in\cK_P]=(2^h-1)/(2^n-1)$.
Since the cosets $\cL+B\kappa/2$ for $\kappa\in\F_2^n$
partition $\cL/2$, we have
\begin{align*}
  \E_P\rho_{\sqrt2/\xi_R}
  (\Lambda_0^*\setminus\{0\})
  &
  \le
  \rho_{\sqrt2/\xi_R}(\cL\setminus\{0\})
  +
  2^{h-n+o(n)}
  \left(
    \rho_{2\sqrt2/\xi_R}(\cL)
    -
    \rho_{\sqrt2/\xi_R}(\cL)
  \right)
  \\
  &
  \le
  2^{(-g(R)+o(1))n}
  +
  2^{(\chi-g_2(R)+o(1))n}.
\end{align*}
The last inequality follows from \cref{eqn:simplerho_g} with
$c=0$ and \cref{lem:parity-source-mass}. Since $\chi<g_2(R)$ and $R>t_0$, Markov's inequality in
\cref{lem:Markov} gives
$\rho_{\sqrt2/\xi_R}(\Lambda_0^*\setminus\{0\})
=2^{-\Omega(n)}<1/2$ except with probability $2^{-\Omega(n)}$.
The last claim follows from the definition and monotonicity of the
smoothing parameter.
\end{proof}



Given this lemma, we can apply \cref{thm:dgs-sampling} to sample the discrete Gaussian on $\Lambda_j$. More precisely, we define $\Gamma:=
  \Lambda_0+\Z x =\Lambda_0\sqcup\Lambda_j$ for $j\neq 0$ and $x\in\Lambda_j$, and run the algorithm from \cref{thm:dgs-sampling} on $\Gamma$.
  Collecting the samples in $\Lambda_j$ works well. This gives the following lemma.

\begin{lemma}\label{lem:source-coset-sampling}
Let $\lambda=\lambda_1(\cL)$,
$\lambda\le d\le(1+1/n)\lambda$, $\chi<g_2(R)$, $R>t_0$,
and $M=2^{O(n)}$.
With probability $1-2^{-\Omega(n)}$ over uniform
$P\in\GL_n(\F_2)$, the following holds for every
$j\in\F_2^h$. There is an algorithm that outputs $M$ vectors or
aborts, and its abort probability is at most $2^{-\Omega(n)}$.
If it does not abort, the joint distribution of its outputs has statistical
distance at most $\exp(-\Omega(n^2))$ from
$D_{\Lambda_j,\xi_R}^M$.
The algorithm takes time
$(M+2^{n/2})2^{o(n)}$ and space $2^{n/2+o(n)}$.
\end{lemma}

\begin{proof}
Note that $\rho_{\sqrt2/\xi_R}
  \left(
    \Lambda_0^*\setminus\{0\}
  \right)
  \le
  2^{-\Omega(n)}$ and $\xi_R>\sqrt2\,\eta_{1/2}(\Lambda_0)$ holds except with probability
$2^{-\Omega(n)}$ over $P$ because of \cref{lem:refined-source-coset-smoothing}.
Fix any such $P$.

For $j=0$, we can directly use \cref{thm:dgs-sampling} $\lceil M/2^{n/2}\rceil$ times to $\Lambda_0$ with
parameter $\xi_R$ and $\kappa=n^2$.
Suppose that $j\ne0$. 
Solve the linear equations $J_P(x)=j$ over $\F_2$ and lift a
solution to obtain $x\in\cL^*$ satisfying $x\in\Lambda_j$. Let
$  \Gamma
  :=
  \Lambda_0+\Z x
  =
  \Lambda_0\sqcup\Lambda_j.$
Since $\Gamma^*\subseteq\Lambda_0^*$, we have $\rho_{\sqrt2/\xi_R}
  \left(
    \Gamma^*\setminus\{0\}
  \right)
  \le
  2^{-\Omega(n)},$
and hence
$\xi_R>\sqrt2\,\eta_{1/2}(\Gamma)$ for large $n$.

Choose arbitrary $y\in\Lambda_0^*\setminus\Gamma^*$.
Poisson summation formula gives
\[
  \left|
    \E_{X\sim D_{\Gamma,\xi_R}}
    e^{2\pi i\ip{X}{y}}
  \right|
  =
  \frac{
    \rho_{1/\xi_R}(\Gamma^*-y)
  }{
    \rho_{1/\xi_R}(\Gamma^*)
  }
  \le
  \rho_{\sqrt2/\xi_R}
  \left(
    \Lambda_0^*\setminus\{0\}
  \right)
  =
  2^{-\Omega(n)}.
\]
The character $e^{2\pi i\ip{X}{y}}$ equals $1$ for every
$X\in\Lambda_0$. Moreover, $2x\in\Lambda_0$, while
$y\notin\Gamma^*$, so $\ip{x}{y}\in\Z+1/2$. Hence the character
equals $-1$ for every $X\in x+\Lambda_0=\Lambda_j$. The above
inequality gives $|1-2p|\le2^{-\Omega(n)}$, where
$p=\Prob_{X\sim D_{\Gamma,\xi_R}}[X\in\Lambda_j]$. It follows that
$p=1/2+O(2^{-\Omega(n)})$.

  We obtain the following sampling algorithm: Use \cref{thm:dgs-sampling} on $\Gamma$ to generate
$4(M+n)$ candidates, with parameter $\xi_R$ and $\kappa=n^2$.
Collect the first $M$ candidates in $\Lambda_j$ and abort if fewer
than $M$ are obtained. The Chernoff bound in
\cref{lem:Chernoff} gives abort probability $2^{-\Omega(n)}$.

For exact samples from $D_{\Gamma,\xi_R}$, conditioning on
$X\in\Lambda_j$ gives exactly $D_{\Lambda_j,\xi_R}$. There are at
most $2^{O(n)}$ calls to
\cref{thm:dgs-sampling}, so their total statistical distance is
still $\exp(-\Omega(n^2))$. Applying the collection algorithm does
not increase statistical distance, and conditioning on its
non-abort event changes it by at most a factor
$1+2^{-\Omega(n)}$.
Finally, the calls can be made sequentially, giving time
$(M+2^{n/2})2^{o(n)}$ and space $2^{n/2+o(n)}$.
\end{proof}

\subsection{Estimating the Hessian by importance sampling}
\label{sec:importance-estimation}

Importance sampling estimates an expectation under a target
distribution using samples from another distribution. Here we sample
from $D_{\Lambda_j,\xi_R}$ for large $\xi_R$ and estimate the
Hessian defined using $D_{\Lambda_j,\xi_r}$ for the smaller target
$\xi_r$. For $X\in\Lambda_j$, let
$w(X):=\rho_{\xi_r}(X)/\rho_{\xi_R}(X)$ and
$\overline w_j:=\E_{X\sim D_{\Lambda_j,\xi_R}}[w(X)]
=\rho_{\xi_r}(\Lambda_j)/\rho_{\xi_R}(\Lambda_j)$.
For every matrix-valued function $H$ for which the expectations
exist,
\begin{align}
    \label{eqn: change_of_weight}
  \E_{X\sim D_{\Lambda_j,\xi_R}}[w(X)H(X)]
  =
  \sum_{X\in \Lambda_j} \frac{\rho_{\xi_R}(X)}{\rho_{\xi_R}(\Lambda_j)}  \frac{\rho_{\xi_r}(X)}{\rho_{\xi_R}(X)} H(X)
  =
  \overline w_j
  \E_{X\sim D_{\Lambda_j,\xi_r}}[H(X)].
\end{align}
Thus the weighted expectation is
$\overline w_j\cG_{r,d,j}(\theta)$ when
$H(X)=-4\pi^2XX^T(-1)^{\theta\cdot V_P(X)}$.


Let $M$ be the number of sample to be determined later.
Consider $n$ independent families
$\{X_{a,1},\ldots,X_{a,M}\}$ of samples from
$D_{\Lambda_j,\xi_R}$ for $a\in [n]$.
 Define the importance-sampling Hessian estimator
\begin{align}\label{eqn:importance-estimator}
  \widehat{\cN}_{a,j}(\theta)
  &:=
  -\frac{4\pi^2}{M}
  \sum_{i=1}^M
  w(X_{a,i})
  X_{a,i}X_{a,i}^T
  (-1)^{\theta\cdot V_P(X_{a,i})}
  \ind_{\{\norm{X_{a,i}}\le C_1\xi_R\sqrt n\}}.
\end{align}
For $p,q\in[n]$, define the $(p,q)$-entry of
$\widehat{\cN}_j(\theta)$ to be the median (i.e., 
$(\lfloor n/2\rfloor+1)$-st smallest) among the corresponding
entries of
$\widehat{\cN}_{1,j}(\theta),\ldots,
\widehat{\cN}_{n,j}(\theta)$. We will show the median estimator $\widehat{\cN}_j(\theta)$ approximates $\overline w_j\cG_{r,d,j}(\theta)$ well for uniform random $P$ and $j$.

We need some notations.
Define $s>0$ by $1/s=2/r-1/R$. Then
\begin{align}\label{eqn:s_identities}
  \frac1{\xi_s^2}
  =
  \frac2{\xi_r^2}-\frac1{\xi_R^2},
  \qquad
  \frac{\rho_{\xi_r}(X)^2}{\rho_{\xi_R}(X)}
  =
  \rho_{\xi_s}(X),
  \qquad
  s=\frac{rR}{2R-r}<r<R.
\end{align}

We also recall $\mu$ from \cref{eqn: mudef}, and choose the number of samples $M$ using $\iota$ defined below:
\begin{align}\label{eqn:importance-loss}
\mu_{r,d}
  =
  \frac{2\pi^2\xi_r^4\lambda^2}
       {\rho_{1/\xi_r}(\cL)}
  e^{-\pi\xi_r^2\lambda^2/4},
  \qquad
  M=\lceil n^6 2^{(\iota(r,R)+2r)n}\rceil\quad\text{for} \quad
  \iota(r,R)
  &:=
  \frac12\log_2\frac{R^2}{r(2R-r)}
\end{align}
where $\mu_{r,d}=\Theta(\xi_r^4\lambda^2 2^{-rn})$ if $\lambda\le d\le(1+1/n)\lambda$ and $s>t_0/2$.

\begin{lemma}\label{lem:importance-weight-bounds}
Let $\lambda=\lambda_1(\cL)$ and
$\lambda\le d\le(1+1/n)\lambda$.
If $s>t_0/2$ and $\chi<1/2+g(r)$, then,
the following hold
with probability
$1-O(n^{-1})$ over uniform random $P\in\GL_n(\F_2)$ and
$j\in\F_2^h$:
\[
  \zeta_{r,j}\ge\frac12,
  \qquad
  \overline w_j
  =
  \left(\frac rR\right)^{n/2}
  \left(
    1+O(2^{-\Omega(n)})
  \right),
  \qquad
  \frac{
    \E_{X\sim D_{\Lambda_j,\xi_R}}[w(X)^2]
  }{
    \overline w_j^2
  }
  \le
  2n2^{\iota(r,R)n}.
\]
\end{lemma}

\begin{proof}
Since $s>t_0/2$ and $s<r<R$,
\cref{cor:working-scale} gives
$\rho_{1/\xi_s}(\cL),
\rho_{1/\xi_r}(\cL),
\rho_{1/\xi_R}(\cL)
=1+2^{-\Omega(n)}$.
Since $g$ is increasing,
\cref{lem:target-coset-flatness} gives
$\zeta_{r,j},\zeta_{R,j}
=1+O(2^{-\Omega(n)})$ simultaneously for every $j$, except with
probability $2^{-\Omega(n)}$ over $P$. In particular,
$\zeta_{r,j}\ge1/2$ for all sufficiently large $n$.

By the definitions of $\overline w_j$ and $\zeta_{t,j}$, and by applying the Poisson summation formula to $\cL^*$, we have
\[
  \overline w_j
  =
  \frac{\zeta_{r,j}}{\zeta_{R,j}}
  \frac{\rho_{\xi_r}(\cL^*)}{\rho_{\xi_R}(\cL^*)}
  =
  \left(\frac rR\right)^{n/2}
  \frac{\rho_{1/\xi_r}(\cL)}{\rho_{1/\xi_R}(\cL)}
  \frac{\zeta_{r,j}}{\zeta_{R,j}}=\left(\frac rR\right)^{n/2} \left(
    1+O(2^{-\Omega(n)})
  \right)
\]
where we use the approximations of $\rho$ and $\zeta$ above.
This proves the estimate of $\overline w_j$ in the
statement.

Fix a choice of $P$ satisfying the preceding conclusions. Since
$2^{-h}\sum_j\zeta_{s,j}=1$ and
$\zeta_{R,j}/\zeta_{r,j}^2
=1+2^{-\Omega(n)}$ uniformly in $j$, the expectation over
uniform $j$ of
$\zeta_{s,j}\zeta_{R,j}/\zeta_{r,j}^2$ is
$1+2^{-\Omega(n)}$. Markov's inequality gives
\[
  \frac{
    \zeta_{s,j}\zeta_{R,j}
  }{
    \zeta_{r,j}^2
  }
  \le n
\]
except with probability $O(n^{-1})$ over uniform $j$.

The definitions of $w(X)$ and
$D_{\Lambda_j,\xi_R}$, together with
\cref{eqn:s_identities}, give
\[
  \E_{X\sim D_{\Lambda_j,\xi_R}}[w(X)^2]
  =
  \frac1{\rho_{\xi_R}(\Lambda_j)}
  \sum_{X\in\Lambda_j}
  \frac{\rho_{\xi_r}(X)^2}{\rho_{\xi_R}(X)}
  =
  \frac{\rho_{\xi_s}(\Lambda_j)}
       {\rho_{\xi_R}(\Lambda_j)}.
\]

Also,
$\overline w_j
=\rho_{\xi_r}(\Lambda_j)/\rho_{\xi_R}(\Lambda_j)$.
Consequently, we can compute ${\E_{X\sim D_{\Lambda_j,\xi_R}}[w(X)^2]}/
       {\overline w_j^2}$ as follows:
\[
  \frac{
    \rho_{\xi_s}(\Lambda_j)
    \rho_{\xi_R}(\Lambda_j)
  }{
    \rho_{\xi_r}(\Lambda_j)^2
  }
  =
  \frac{
    \rho_{\xi_s}(\cL^*)
    \rho_{\xi_R}(\cL^*)
  }{
    \rho_{\xi_r}(\cL^*)^2
  }
  \frac{
    \zeta_{s,j}\zeta_{R,j}
  }{
    \zeta_{r,j}^2
  }
  =
  \left(
    \frac{sR}{r^2}
  \right)^{n/2}
  \frac{
    \rho_{1/\xi_s}(\cL)
    \rho_{1/\xi_R}(\cL)
  }{
    \rho_{1/\xi_r}(\cL)^2
  }
  \frac{
    \zeta_{s,j}\zeta_{R,j}
  }{
    \zeta_{r,j}^2
  }.
\]
Here the first equality follows by writing
$\rho_{\xi_t}(\Lambda_j)
=2^{-h}\zeta_{t,j}\rho_{\xi_t}(\cL^*)$, and the second equality is obtained by
applying Poisson summation to $\cL^*$. 
Plugging $s=rR/(2R-r)$ and 
$R/\sqrt{r(2R-r)}=2^{\iota(r,R)}$ gives the multiplicative factor $2^{\iota(r,R)n}$.
The last two factors are at most
$2$ and $n$, respectively, for all sufficiently large $n$.
This proves the result.
\end{proof}

\begin{lemma}\label{lem:importance-coset-estimation}
Let $\lambda=\lambda_1(\cL)$ and
$\lambda\le d\le(1+1/n)\lambda$. If $s>t_0/2$ and
$\chi<1/2+g(r)$, then
it holds simultaneously for every $\theta\in\F_2^\ell$:
\[
  \max_{\theta\in\F_2^\ell}
  \opnorm{
    \widehat{\cN}_j(\theta)
    -
    \overline w_j\cG_{r,d,j}(\theta)
  }
  =
  O\left(
    \frac{\overline w_j\mu_{r,d}}n
  \right)
\]
with probability $1-O(n^{-1})$ over
$P,j$, and the choices of $X_{a,i}$.
\end{lemma}

\begin{proof}

Suppose that the conclusions of
\cref{lem:importance-weight-bounds} hold, and fix the corresponding
pair $P,j$. For every $a$, the expectation of
$\widehat{\cN}_{a,j}(\theta)$ is independent of $a$. From
\cref{eqn:importance-estimator,eqn: change_of_weight}, we have
\[
  \E[
    \widehat{\cN}_{a,j}(\theta)
  ]
  =
  -4\pi^2\overline w_j
  \E_{X\sim D_{\Lambda_j,\xi_r}}
  \left[
    XX^T(-1)^{\theta\cdot V_P(X)}
    \ind_{\{\norm{X}\le C_1\xi_R\sqrt n\}}
  \right].
\]
Consequently,
\begin{align*}
  \widehat{\cN}_j(\theta)
  -
  \overline w_j\cG_{r,d,j}(\theta)
  &=
  \left(
    \widehat{\cN}_j(\theta)
    -
    \E[
      \widehat{\cN}_{a,j}(\theta)
    ]
  \right)+
  4\pi^2\overline w_j
  \E_{X\sim D_{\Lambda_j,\xi_r}}
  \left[
    XX^T(-1)^{\theta\cdot V_P(X)}
    \ind_{\{\norm{X}>C_1\xi_R\sqrt n\}}
  \right].
\end{align*}We bound the two terms separately.

We first bound the first term by applying Chebyshev's inequality
to each entry. Fix $a,\theta$ and the entry indices $p,q$. Write $(\widehat{\cN}_{a,j}(\theta))_{p,q}
=M^{-1}\sum_{i=1}^M Y_i$ by defining
\[
  Y_i
  :=
  -4\pi^2w(X_{a,i})
  (X_{a,i})_p(X_{a,i})_q
  (-1)^{\theta\cdot V_P(X_{a,i})}
  \ind_{\{\norm{X_{a,i}}\le C_1\xi_R\sqrt n\}}.
\]
This gives the following upper bound
\[
  \E\left[
    \left(
      \widehat{\cN}_{a,j}(\theta)
      -
      \E[
        \widehat{\cN}_{a,j}(\theta)
      ]
    \right)_{p,q}^2
  \right]
  =
  \frac1{M^2}
  \sum_{i=1}^M
  \E\left[
    (Y_i-\E[Y_i])^2
  \right]
  \le
  \frac1M\E[Y_1^2].
\]
where the cross terms between
$Y_i-\E[Y_i]$ and $Y_k-\E[Y_k]$ vanish for $i\ne k$ since the $Y_i$ are independent.
Consequently,
\[
  \E\left[
    \left(
      \widehat{\cN}_{a,j}(\theta)
      -
      \E[
        \widehat{\cN}_{a,j}(\theta)
      ]
    \right)_{p,q}^2
  \right]
  \le
  \frac1M
  \E_{X\sim D_{\Lambda_j,\xi_R}}
  \left[
    w(X)^2(4\pi^2X_pX_q)^2
    \ind_{\{\norm{X}\le C_1\xi_R\sqrt n\}}
  \right].
\]
When the indicator is nonzero,
$(4\pi^2X_pX_q)^2=O(\xi_R^4 n^2)=O(\xi_r^8\lambda^4)$. Hence
\[
  \E\left[
    \left(
      \widehat{\cN}_{a,j}(\theta)
      -
      \E[
        \widehat{\cN}_{a,j}(\theta)
      ]
    \right)_{p,q}^2
  \right]
  =
  O\left(\frac{\xi_r^8\lambda^4}M
  \E_{X\sim D_{\Lambda_j,\xi_R}}[w(X)^2]\right)=O\left(
    \frac{
      n\overline w_j^2\xi_r^8\lambda^4
      2^{\iota(r,R)n}
    }M
  \right) = O\left(
    \frac{\overline w_j^2\mu_{r,d}^2}{n^5}
  \right)
\]
where we use $\E[w(X)^2]\le
2n\overline w_j^2 2^{\iota(r,R)n}$ from \cref{lem:importance-weight-bounds} and the definition of $\mu_{r,d}$ and $M$ from \cref{eqn:importance-loss}.

Chebyshev's inequality now gives
\[
  \Prob\left[
    \left|
      \left(
        \widehat{\cN}_{a,j}(\theta)
        -
        \E[
          \widehat{\cN}_{a,j}(\theta)
        ]
      \right)_{p,q}
    \right|
    >
    \frac{\overline w_j\mu_{r,d}}{n^2}
  \right]
  =
  O(n^{-1}).
\]
Thus each of the $n$ independent estimates of the $(p,q)$-entry
has the required accuracy with probability $1-O(n^{-1})$. If more
than half of them have this accuracy, then their median has the same
accuracy. By \cref{lem:majority}, the probability that the median
does not have this accuracy is $2^{-\Omega(n\log n)}$.
A union bound over $\theta\in\F_2^\ell$ and $p,q\in[n]$, together
with $\opnorm{A}\le n\max_{p,q}|A_{p,q}|$, therefore gives
\[
  \max_{\theta\in\F_2^\ell}
  \opnorm{
    \widehat{\cN}_j(\theta)
    -
    \E[
      \widehat{\cN}_{a,j}(\theta)
    ]
  }
  =
  O\left(
    \frac{\overline w_j\mu_{r,d}}n
  \right)
\]
except with probability $2^{-\Omega(n\log n)}$.

It remains to bound the second term.
Since $\zeta_{r,j}\ge1/2$ and its definition $\zeta_{r,j}=2^h \rho_{\xi_r}(\Lambda_j)/\rho_{\xi_r}(\cL^*)$, we have
$\rho_{\xi_r}(\Lambda_j)
\ge2^{-h-1}\rho_{\xi_r}(\cL^*)$. Since $\xi_R>\xi_r$,
\cref{lem:dgs-tail} shows that the operator norm of the second term
in the decomposition is at most
$4\pi^2\overline w_j2^{h+1}\xi_r^2n2^{-n}$.
Dividing this bound by $\overline w_j\mu_{r,d}/n$ gives an upper bound
$O(n2^{h-(1-r)n})=2^{-\Omega(n)}$, because $h\le n/2$ and
$r<1/4$. Therefore, the second term is much smaller:
\[
  4\pi^2\overline w_j
  \opnorm{
    \E_{X\sim D_{\Lambda_j,\xi_r}}
    \left[
      XX^T(-1)^{\theta\cdot V_P(X)}
      \ind_{\{\norm{X}>C_1\xi_R\sqrt n\}}
    \right]
  }
  =
  o(\overline w_j\mu_{r,d}/n).
\]

Combining the two terms proves
the lemma. The total failure probability is the summand of $O(n^{-1})$ from
\cref{lem:importance-weight-bounds} and the
additional failure probability
$2^{-\Omega(n\log n)}$.
\end{proof}

\subsection{Recovering a shortest vector}
\label{sec:importance-recovery}
The following lemma is almost identical to \cref{lem:affine-coset-recovery}, except that we have a normalizing factor $\overline{w}_j$ and we use \cref{lem:parity-rank-one} instead of \cref{lem:shortest-class-hessian}.

\begin{lemma}\label{lem:importance-coset-recovery}
Let $\lambda=\lambda_1(\cL)\le d\le(1+1/n)\lambda$.
Let $v\in\cL$ satisfy $\norm{v}=\lambda$, and let
$u_*\in\F_2^n$ satisfy $v\in Bu_*+2\cL$.
Choose independent uniform $P\in\GL_n(\F_2)$ and
$j\in\F_2^h$, and write $Pu_*=(\alpha_*,\theta_*)$.
Assume that $r\ge9/50$, $\chi<1/2+g(r)$,
  $\chi-1-\min\{g(r),2g(r)\}+2r<0$ and that the conclusion of
\cref{lem:importance-coset-estimation} holds.
  The following holds 
with probability $1-O(n^{-1})$ over $P$ and $j$: 
Let $q_+$ and $q_-$ be unit eigenvectors corresponding to the
largest and smallest eigenvalues of
$\widehat{\cN}_j(\theta_*)$, respectively.
Then there are $s\in\{+,-\}$ and $\tau\in\{-1,1\}$ such that
\[
  \norm{q_s-\tau v/\lambda}
  =
  O(n^{-1/2}).
\]
Consequently, the $n^{-1/3}$-$\BDD$ query at $dq_s$ returns
$\tau v$ for all sufficiently large $n$.
\end{lemma}

\begin{proof}
Put $\widehat v=v/\lambda$ and
$\sigma=(-1)^{\alpha_*\cdot j}$. By
\cref{lem:parity-rank-one},
$\cA_{r,d}(u_*)=\mu_{r,d}\widehat v\widehat v^T+\cR$, where
$\Frob{\cR}\le\mu_{r,d}2^{-\Omega(n)}$.
Let
$S_{P,j}:=\sum_{\alpha\ne\alpha_*}
(-1)^{j\cdot(\alpha-\alpha_*)}
\cA_{r,d}(P^{-1}(\alpha,\theta_*))$.
The calculation in the proof of
\cref{lem:affine-coset-recovery}, with $t$ replaced by $r$,
together with \cref{lem:hessian-square-mass}, gives
\[
  \E_{P,j}\Frob{S_{P,j}}^2
  \le
  \mu_{r,d}^2
  2^{(\chi-1-\min\{g(r),2g(r)\}+2r+o(1))n}.
\]
The exponent is negative by assumption. Hence Markov's inequality
gives $\Frob{S_{P,j}}\le\mu_{r,d}2^{-\Omega(n)}$, except with
probability $2^{-\Omega(n)}$ over $P,j$.

By \cref{lem:target-coset-flatness},
$\zeta_{r,j}=1+O(2^{-\Omega(n)})$, except with probability
$2^{-\Omega(n)}$ over $P$. Applying
\cref{lem:affine-coset-hessian} and substituting the expressions for
$\cA_{r,d}(u_*)$ and $S_{P,j}$, and then applying
\cref{lem:importance-coset-estimation}, gives, for some $a\in\R$,
\[
  \sigma\widehat{\cN}_j(\theta_*)
  =
  a\Id+
  \frac{\overline w_j\mu_{r,d}}{\zeta_{r,j}}
  \widehat v\widehat v^T+E,
  \qquad
  \opnorm{E}
  =
  O\left(\frac{\overline w_j\mu_{r,d}}n\right).
\]

If $\sigma=1$, let $q=q_+$, and otherwise let $q=q_-$. The remainder proof is identical to \cref{lem:affine-coset-recovery} except that we use $\overline{w}_j \mu_{r,d}/\zeta_{r,j}$ instead of $\mu_{t,d}/\zeta_{t,j}$.
\end{proof}

\subsection{The importance-sampling algorithm}
\label{sec:importance-algorithm}
This section presents the algorithm using importance sampling.
The time complexity is determined by the
$2^{(\iota(r,R)+2r)n+o(n)}$ samples used in
\cref{lem:importance-coset-estimation} and the
$2^{\ell+o(n)}=2^{(1-\chi)n+o(n)}$ matrices, e.g., $\widehat{\cN}_j(\theta)$, indexed by
$\theta\in\F_2^\ell$. The space complexity is determined by the
samples and these matrices.

We choose the following parameters
\[
  r=0.2222355,\qquad
  R=0.400613,\qquad
  \chi=0.3961331,
\]
which gives $s=\frac{rR}{2R-r}=0.1537683785\ldots$ and $\iota(r,R)
  =
  \frac12\log_2\frac{R^2}{r(2R-r)}
  =
  0.1593947534\ldots$.
A direct numerical calculation shows that all the conditions of 
\cref{lem:source-coset-sampling,lem:importance-coset-estimation,lem:parity-rank-one,lem:importance-coset-recovery}, and
\begin{align}\label{eqn:importance-exponent}
  \max\left\{
    \frac12,
    \iota(r,R)+2r,
    1-\chi
  \right\}
  <
  0.603867.
\end{align}

\begin{remark}\label{rem: nonopt}
Without the refined bound in
\cref{lem:parity-source-mass}, the source-coset sampling argument
requires $\chi<g(R)$. Choose
$r=0.222275$, $R=0.400765$, and
$\chi=0.395965$. Then
\[
  g(R)=0.395965\ldots>\chi,
  \qquad
  \iota(r,R)+2r=0.604035\ldots,
  \qquad
  1-\chi=0.604035.
\]
These values satisfy the other assumptions of the algorithm.
Therefore, using $g(R)$ instead of $g_2(R)$ gives time and space
$2^{0.60404n+o(n)}$.
\end{remark}

\begin{algorithmblock}{Importance-sampling affine-coset Hessian
$\SVP$}
\label{alg:importance-full-space}
  \item Construct the preprocessing $n^{-1/3}$-$\BDD$ data and the
        scale grid from \cref{alg:direct-exhaustive}.

  \item At every scale $d$, choose independent uniform
        $P\in\GL_n(\F_2)$ and $j\in\F_2^h$. Put
        $M=\lceil n^6 2^{(\iota(r,R)+2r)n}\rceil$ and apply the
        sampling algorithm from \cref{lem:source-coset-sampling}.
        Abort the scale unless $nM$ samples from
        $D_{\Lambda_j,\xi_R}$ are obtained, and divide them into
        $n$ families
        $\{X_{a,1},\ldots,X_{a,M}\}$, $a\in[n]$.

  \item For every $a\in[n]$, construct the matrix-valued array
        \[
          A_a(y)
          :=
          -\frac{4\pi^2}{M}
          \sum_{\substack{1\le i\le M\\V_P(X_{a,i})=y}}
          w(X_{a,i})X_{a,i}X_{a,i}^T
          \ind_{\{\norm{X_{a,i}}\le C_1\xi_R\sqrt n\}},
          \qquad
          y\in\F_2^\ell.
        \]

  \item Apply the matrix-valued Walsh-Hadamard transform to each
        $A_a$. For every $\theta\in\F_2^\ell$, the output satisfies
        \begin{align}\label{eqn:importance-wht}
          \sum_{y\in\F_2^\ell}
          (-1)^{\theta\cdot y}A_a(y)
          =
          \widehat{\cN}_{a,j}(\theta).
        \end{align}
        Define $\widehat{\cN}_j(\theta)$ by taking the median of the
        corresponding entries of
        $\widehat{\cN}_{1,j}(\theta),\ldots,
        \widehat{\cN}_{n,j}(\theta)$.

  \item For every $\theta\in\F_2^\ell$, compute unit eigenvectors
        $q_+$ and $q_-$ corresponding to the largest and smallest
        eigenvalues of $\widehat{\cN}_j(\theta)$ and query $\BDD$ at
        $dq_\pm$. Store a nonzero output in $\cL$. Keep only the
        shortest vector found and discard the other vectors and
        temporary data.

  \item Return the shortest vector over all scales.
\end{algorithmblock}

\begin{theorem}\label{thm:importance-full-space}
\cref{alg:importance-full-space} solves Search-$\SVP$ with constant
success probability in time
$2^{0.60387n+o(n)}$
and space
$2^{0.60387n+o(n)}$.
\end{theorem}

\begin{proof}
Suppose
$\lambda_1(\cL)\le d\le(1+1/n)\lambda_1(\cL)$ and let $v$ be a
shortest vector. We can ignore the other $d$'s as before. Choose
$u_*\in\F_2^n$ such that $v\in Bu_*+2\cL$, and write
$Pu_*=(\alpha_*,\theta_*)$.

By \cref{lem:source-coset-sampling}, the algorithm obtains samples
whose joint distribution is
$\exp(-\Omega(n^2))$-close to independent samples from
$D_{\Lambda_j,\xi_R}$. By \cref{eqn:importance-wht}, the
Walsh-Hadamard transforms compute
$\widehat{\cN}_{a,j}(\theta)$ for every
$a\in[n]$ and $\theta\in\F_2^\ell$. Their entrywise medians are
therefore the matrices $\widehat{\cN}_j(\theta)$ defined in
\cref{sec:importance-estimation}.

The conclusion of \cref{lem:importance-coset-estimation} holds
simultaneously for every $\theta\in\F_2^\ell$, and
\cref{lem:importance-coset-recovery} shows that one of the two
$\BDD$ queries corresponding to $\theta_*$ returns one of
$\pm v$. The statistical distance in
\cref{lem:source-coset-sampling} changes the success probability by
at most $\exp(-\Omega(n^2))$. Since every stored vector is verified
to be a nonzero vector in $\cL$, the shortest stored vector has norm
$\lambda_1(\cL)$.

We analyze the complexity. Generating and processing the samples
takes
$2^{(\max\{1/2,\iota(r,R)+2r\}+o(1))n}$ time by
\cref{lem:source-coset-sampling}. Constructing the arrays takes
$nM\poly(n)$ time. Their Walsh-Hadamard transforms
(\cref{lem:matrix-wht}), the entrywise medians, the eigenvector
computations, and the $\BDD$ queries take
$2^{\ell+o(n)}=2^{(1-\chi)n+o(n)}$ time. Hence the time exponent is $0.60387$ as in \cref{eqn:importance-exponent}.

The Gaussian samples use
$2^{(\iota(r,R)+2r)n+o(n)}$ space, the arrays and their transforms
use $2^{(1-\chi)n+o(n)}$ space, and the sampling algorithm uses
$2^{n/2+o(n)}$ space. Thus the space exponent is bounded by the
same maximum in \cref{eqn:importance-exponent}. This proves the
claimed time and space bounds.
\end{proof}

\subsection{Reducing the space by sparsification}
\label{sec:importance-small-space}
We use the same parameters as in the previous section. Choose independent uniform
$P\in\GL_n(\F_2)$ and $j\in\F_2^h$. The algorithm in
\cref{sec:importance-algorithm} stores $n$ arrays indexed by
$\F_2^\ell$ of size $2^\ell = 2^{0.60387n+o(n)}$. 
To reduce the space, we randomly and sparsely select some of the
samples before applying the Walsh-Hadamard transform, and they are only used in the later steps.

More precisely, we do the following.
For $X\in\Lambda_j$, define
\begin{align}\label{eqn:selection-probability}
  \pi(X)
  :=
  \min\left\{
    1,
    \frac{w(X)}
    {2^{\iota(r,R)n}(r/R)^{n/2}}
  \right\}.
\end{align}
Note that all terms are efficiently computable.
Conditional on each sampled $X$, independently choose
$Z\in\{0,1\}$ such that
$\Pr[Z=1\mid X]=\pi(X)$ and
$\Pr[Z=0\mid X]=1-\pi(X)$.
We use the following:
For every function $f$ for which the expectation exists,
\begin{align}
    \label{eqn: Zpiremove}
  \E_{X,Z}\left[
    \frac{Zf(X)}{\pi(X)}
  \right]
  =
  \E_X\left[
    \E\left[
      \left.
      \frac{Zf(X)}{\pi(X)}
      \right|X
    \right]
  \right]
  =
  \E_X[f(X)].
\end{align}

Let $M=\lceil n^6 2^{(\iota(r,R)+2r)n}\rceil$, and take $n$
independent families
$\{X_{a,1},\ldots,X_{a,M}\}$, $a\in[n]$, of samples from
$D_{\Lambda_j,\xi_R}$. Conditional on these samples, let the
$Z_{a,i}$ be independent random variables such that
$Z_{a,i}=1$ with probability $\pi(X_{a,i})$ and $Z_{a,i}=0$
otherwise. Define the sparsified estimator
\begin{align}\label{eqn:sparse-importance-estimator}
  \widehat{\cN}^{\rm sp}_{a,j}(\theta)
  :=
  -\frac{4\pi^2}{M}
  \sum_{i=1}^M
  \frac{Z_{a,i}w(X_{a,i})}{\pi(X_{a,i})}
  X_{a,i}X_{a,i}^T
  (-1)^{\theta\cdot V_P(X_{a,i})}
  \ind_{\{\norm{X_{a,i}}\le C_1\xi_R\sqrt n\}}.
\end{align}
Because of \cref{eqn: Zpiremove}, the expectation of this sparsification is the same as the expectation of the original \cref{eqn:importance-estimator}.
For each matrix entry, define
$\widehat{\cN}^{\rm sp}_j(\theta)$ by taking the
median among the corresponding
entries of
$\widehat{\cN}^{\rm sp}_{1,j}(\theta),\ldots,
\widehat{\cN}^{\rm sp}_{n,j}(\theta)$.

The factor $1/\pi(X_{a,i})$ compensates for selecting a sample with
probability $\pi(X_{a,i})$: conditional on $X_{a,i}$, the
expectation of $Z_{a,i}w(X_{a,i})/\pi(X_{a,i})$ is $w(X_{a,i})$.

\begin{lemma}\label{lem:sparse-importance-estimation}
Let $\lambda=\lambda_1(\cL)$ and
$\lambda\le d\le(1+1/n)\lambda$. Then, the following holds for all $\theta\in \F_2^\ell$:
\[
  \max_{\theta\in\F_2^\ell}
  \opnorm{
    \widehat{\cN}^{\rm sp}_j(\theta)
    -
    \overline w_j\cG_{r,d,j}(\theta)
  }
  =
  O\left(
    \frac{\overline w_j\mu_{r,d}}n
  \right)
\]
with probability $1-O(n^{-1})$ over $P,j$, the samples $X_{a,i}$, and the
variables $Z_{a,i}$.
\end{lemma}

\begin{proof}
Fix $P,j$ for which the conclusion of
\cref{lem:importance-weight-bounds} holds, which happens with
probability $1-O(n^{-1})$.
For fixed $a$, let $\E_{X_a}$ denote expectation over the
independent samples
$X_{a,1},\ldots,X_{a,M}\sim D_{\Lambda_j,\xi_R}$.
Let $\E_{X_a,Z_a}$ additionally include the conditionally
independent choices
$Z_{a,i}\mid X_{a,i}\sim\operatorname{Bernoulli}
(\pi(X_{a,i}))$.

By \cref{eqn: Zpiremove},
$\E_{X_a,Z_a}[
\widehat{\cN}^{\rm sp}_{a,j}(\theta)]
=
\E_{X_a}[
\widehat{\cN}_{a,j}(\theta)]$.
Consequently,
\[
  \widehat{\cN}^{\rm sp}_j(\theta)
  -\overline w_j\cG_{r,d,j}(\theta)
  =
  \widehat{\cN}^{\rm sp}_j(\theta)
  -\E_{X_a,Z_a}[
    \widehat{\cN}^{\rm sp}_{a,j}(\theta)]
  +\E_{X_a}[
    \widehat{\cN}_{a,j}(\theta)]
  -\overline w_j\cG_{r,d,j}(\theta).
\]
The proof of \cref{lem:importance-coset-estimation} shows that the
operator norm of
$\E_{X_a}[\widehat{\cN}_{a,j}(\theta)]
-\overline w_j\cG_{r,d,j}(\theta)$ is
$o(\overline w_j\mu_{r,d}/n)$. We bound the first difference using
the same concentration argument.

By the definition of $\pi(X)$ and $1/\pi(X)=\max(1,2^{\iota(r,R)n}(r/R)^{n/2}/w(X)) \le 1+2^{\iota(r,R)n}(r/R)^{n/2}/w(X)$,
\[
  \E_{Z\mid X}\left[
    \left(
      \frac{Zw(X)}{\pi(X)}
    \right)^2
  \right]
  =
  \frac{w(X)^2}{\pi(X)}
  \le
  w(X)^2
  +
  2^{\iota(r,R)n}(r/R)^{n/2}w(X),
\]
where $Z\mid X\sim\operatorname{Bernoulli}(\pi(X))$.
In the proof of 
\cref{lem:importance-coset-estimation}, we proved that $\Frob{4\pi^2XX^T}=O(\xi_r^4\lambda^2)$ if $\{\norm{X}\le C_1\xi_R\sqrt n\}$. Hence
\cref{lem:importance-weight-bounds}, together with
$  \E_{X\sim D_{\Lambda_j,\xi_R}}[w(X)]
  =
  \overline w_j=(r/R)^{n/2}(1+O(2^{-\Omega(n)}))$, gives
\[
  \E_{\substack{
    X\sim D_{\Lambda_j,\xi_R}\\
    Z\mid X\sim\operatorname{Bernoulli}(\pi(X))
  }}
  \left[
    \left(
      \frac{Zw(X)}{\pi(X)}
    \right)^2
    \Frob{4\pi^2XX^T}^2
    \ind_{\{\norm{X}\le C_1\xi_R\sqrt n\}}
  \right]
  =
  O\left(
    n\overline w_j^2\xi_r^8\lambda^4
    2^{\iota(r,R)n}
  \right).
\]

For fixed $a,\theta,p,q$, the estimator entry is the average of
$M$ independent summands. Therefore,
\[
  \E_{X_a,Z_a}\left[
    \left(
      \widehat{\cN}^{\rm sp}_{a,j}(\theta)
      -
      \E_{X_a,Z_a}[
        \widehat{\cN}^{\rm sp}_{a,j}(\theta)
      ]
    \right)_{p,q}^2
  \right]
  =
  O\left(
    \frac{\overline w_j^2\mu_{r,d}^2}{n^5}
  \right)
\]
by the definitions of $M$ and $\mu_{r,d}$ in
\cref{eqn:importance-loss}. This is the same bound used in the
proof of \cref{lem:importance-coset-estimation}. Applying its
Chebyshev, median, and union-bound arguments proves the required
operator-norm bound simultaneously for every
$\theta\in\F_2^\ell$. The failure probability in
\cref{lem:importance-weight-bounds} is $O(n^{-1})$.
\end{proof}

The next lemma bounds the number of samples for which $Z_{a,i}=1$.

\begin{lemma}\label{lem:selected-sample-count}
Under the hypotheses of
\cref{lem:sparse-importance-estimation}, with probability
$1-O(n^{-1})$ over $P,j$, the samples $X_{a,i}$, and the variables
$Z_{a,i}$, the number of pairs $(a,i)$ satisfying $Z_{a,i}=1$ and
$\norm{X_{a,i}}\le C_1\xi_R\sqrt n$ is at most
$n^8 2^{2rn}$.
\end{lemma}

\begin{proof}
Since
$\pi(X)\le
w(X)/(2^{\iota(r,R)n}(r/R)^{n/2})$,
\cref{lem:importance-weight-bounds} gives
\[
  \E_{X\sim D_{\Lambda_j,\xi_R}}[\pi(X)]
  \le
  \frac{\overline w_j}
  {2^{\iota(r,R)n}(r/R)^{n/2}}
  =
  2^{-\iota(r,R)n}
  \left(1+O(2^{-\Omega(n)})\right).
\]
The expected number of pairs satisfying $Z_{a,i}=1$ is therefore
$O(nM2^{-\iota(r,R)n})=O(n^7 2^{2rn})$. The norm condition can
only decrease this number. Markov's inequality proves the lemma.
\end{proof}

\subsection{Space-efficient implementation}
\label{sec:algorithm-small-space}
Now we introduce an algorithm with the reduced space complexity. For every importance-weighted
sample, we choose the random variable $Z$ using
\cref{eqn:selection-probability} and store the sample only if $Z=1$.
The division by $\pi(X)$ preserves the expectation of its
contribution. By \cref{lem:selected-sample-count}, the number of
stored samples is $2^{2rn+o(n)}$. 

We choose the same parameters $r,R,\chi$ as before, satisfying the conditions of all the lemmas and \cref{eqn:importance-exponent}.
The time complexity remains
$2^{0.60387n+o(n)}$. Generating and selecting the samples takes
$2^{(\iota(r,R)+2r)n+o(n)}$ time, and all the Walsh-Hadamard transforms and
the examination of all matrices take
$2^{\ell+o(n)}=2^{(1-\chi)n+o(n)}$ time.
The space complexity becomes $2^{n/2+o(n)}$, because the stored
samples use $2^{2rn+o(n)}$ space for $2r<1/2$, and only
$2^{\lfloor n/2\rfloor}$ matrices are processed at a time.

\begin{algorithmblock}{Space-efficient importance-sampling $\SVP$}
\label{alg:best-svp}
  \item Follow the first two steps of
        \cref{alg:importance-full-space}, but process each sample
        $X_{a,i}$ immediately and then discard it. If
        $\norm{X_{a,i}}\le C_1\xi_R\sqrt n$, choose $Z_{a,i}$ as in
        \cref{eqn:selection-probability}. When $Z_{a,i}=1$, write
        \[
          V_P(X_{a,i})=(k'_{a,i},k''_{a,i})
          \in\F_2^b\times\F_2^{\ell-b},
          \qquad b=\lfloor n/2\rfloor,
        \]
        and store $(a,k'_{a,i},k''_{a,i},W_{a,i})$, where
$          W_{a,i}
          :=
          -\frac{4\pi^2w(X_{a,i})}{M\pi(X_{a,i})}
          X_{a,i}X_{a,i}^T.$
        Abort the scale if more than $n^8 2^{2rn}$ tuples are stored.

  \item For every $\theta''\in\F_2^{\ell-b}$ and $a\in[n]$, form
        \[
          A_{a,\theta''}(x)
          :=
          \sum_{\substack{
            (a,k',k'',W)\text{ stored}\\k'=x
          }}
          (-1)^{\theta''\cdot k''}W,
          \qquad x\in\F_2^b.
        \]
        Apply the matrix-valued Walsh-Hadamard transform and take
        the entrywise median of the $n$ outputs for each
        $\theta'\in\F_2^b$. Perform the eigenvector and $\BDD$ steps
        of \cref{alg:importance-full-space}, and discard the arrays
        before proceeding to the next $\theta''$.

  \item Keep the shortest verified vector over all scales. Repeat a
        sufficiently large constant number of times and return the
        shortest vector found.
\end{algorithmblock}

\begin{theorem}\label{thm:importance-small-space}
\cref{alg:best-svp} solves Search-$\SVP$
with success probability at least $2/3$ in expected time
$2^{0.60387n+o(n)}$ and space
$2^{n/2+o(n)}$.
\end{theorem}

\begin{proof}
The correctness follows from the proof of
\cref{thm:importance-full-space}, replacing
\cref{lem:importance-coset-estimation} by
\cref{lem:sparse-importance-estimation}.
The additional abort probability is $O(n^{-1})$ by
\cref{lem:selected-sample-count}.

The samples are discarded after the selection step, so only
$2^{2rn+o(n)}$ tuples are stored. For each fixed $\theta''$, the
algorithm takes
$\left(2^{2rn}+2^b\right)2^{o(n)}$ time and
$2^{b+o(n)}$ space. Therefore, since $2r<1/2$ and
$b=\lfloor n/2\rfloor$, the total time is
$  2^{\ell-b}\left(2^{2rn}+2^b\right)2^{o(n)}
  =
  2^{(1-\chi)n+o(n)}.$
Together with the sampling time, this is
$2^{0.60387n+o(n)}$. The selected tuples, the arrays for one fixed
$\theta''$, the discrete Gaussian sampler, and the preprocessing
$\BDD$ data use $2^{n/2+o(n)}$ space.
A constant number of repetitions gives success probability at least
$2/3$.
\end{proof}

\subsection{Quantum algorithm}
We describe the quantum version of \cref{alg:best-svp} based on the quantum minimum finding algorithm \cite{DH96}. The application is rather straightforward: We search for $\theta'' \in \F_2^{\ell-b}$ that minimizes the shortest vector found by the preprocessing BDD algorithm. This reduces the complexity from $2^{\ell - b+o(n)} (2^{2rn}+2^b)$ to $2^{(\ell - b)/2+o(n)} (2^{2rn}+2^b)$.

In the quantum version, the collected tuples (of size $2^{2rn+o(n)}$) and the preprocessing BDD advice (of size $2^{o(n)}$) are stored in QRAM, which will be of size $2^{2rn+o(n)}$ at total.
The overall steps for the Hessian estimation, eigenvector evaluations, and the BDD queries can be done reversibly.
This gives the coherent oracle required by quantum minimum finding with only a $2^{o(n)}$ multiplicative overhead.

To maximize the speedup, we also choose $b=\lfloor 2rn\rfloor$ and other parameters by
\[
  r=0.180182,\qquad
  R=0.340429,\qquad
  \chi=0.278262.
\]
giving $2^{0.54106n+o(n)}$ time and $2^{0.5n+o(n)}$ space.
\begin{theorem}\label{thm:quantum-svp}
There is a quantum algorithm that solves
Search-$\SVP$ with success probability at least $2/3$ in expected
time
$  2^{0.54106n+o(n)}$
and space
$  2^{n/2+o(n)}.
$ The QRAM size is $2^{0.36036n+o(n)}$ and
$2^{0.36036n+o(n)}$ qubits.
\end{theorem}





  \bibliographystyle{plain}
  \bibliography{abbrev3,crypto,ref}

\end{document}